\documentclass[letterpaper,journal]{IEEEtran}
\usepackage{amsmath,amsfonts,amsthm}
\usepackage{amssymb}
\usepackage{algorithmic}
\usepackage{algorithm}
\usepackage{array}
\usepackage[caption=false,font=normalsize,labelfont=sf,textfont=sf]{subfig}
\usepackage{textcomp}
\usepackage{stfloats}
\usepackage{url}
\usepackage{verbatim}
\usepackage{graphicx}
\usepackage{cite}
\usepackage{xcolor}
\usepackage{booktabs} 
\usepackage{multirow}
\usepackage{mathtools}
\newtheorem{theorem}{Theorem}
\newtheorem{lemma}{Lemma}
\begin{document}

\title{Joint Beamforming and RBG Scheduling in Multi-Cell
MU-MIMO Networks: A Bounded Deep Unfolding Approach}
\author{Jiansheng~Li,
Shuqi~Chai,
Fan~Xu,
Tian Ding,
Kaiming~Shen,
Guangxu~Zhu,
Junting~Chen

\thanks{\textit{Corresponding author: Shuqi Chai.}}
\thanks{Jiansheng Li, Kaiming Shen, and Junting Chen are with the School of Science and Engineering, The Chinese University of Hong Kong, Shenzhen, 518172, China 
(e-mail: jianshengli@link.cuhk.edu.cn; shenkaiming@cuhk.edu.cn; juntingchen@cuhk.edu.cn).}%
\thanks{Shuqi Chai, Tian Ding and Guangxu Zhu are with the Shenzhen Research Institute of Big Data, Shenzhen 518172, China 
(e-mail: schai@sribd.cn; dingtian@sribd.cn; gxzhu@sribd.cn).}%
\thanks{Fan Xu is with the School of Electronic Information Engineering, Tongji University, Shanghai 201804, China 
(e-mail: xxiaof999@tongji.edu.cn).}%
}

\maketitle

\begin{abstract}
This paper investigates the joint resource block group (RBG) scheduling and beamforming optimization problem for weighted sum-rate (WSR) maximization in multi-cell multiuser multiple-input multiple-output (MU-MIMO) downlink networks. While the Fast Fractional Programming (FastFP) framework provides a reliable model-driven solution, it suffers from conservative continuous beamforming updates and prohibitive computational overhead during the discrete RBG matching phase. To address these bottlenecks, we propose a joint deep unfolding framework comprising two core modules: P-Net and K-Net. Specifically, P-Net learns an adaptive relaxation factor along the FastFP direction, strictly bounded within an ascent-preserving interval to accelerate convergence while retaining stationary-point guarantees. Meanwhile, K-Net learns a long-horizon priority policy to guide a low-complexity greedy assignment, maintaining high assignment quality while bypassing the computationally expensive Hungarian matching. Both networks leverage analytical algorithmic priors and utilize recurrent parameter sharing, enabling flexible inference beyond the training horizon. Extensive simulations demonstrate that the proposed joint framework achieves higher WSR and faster execution times than conventional model-driven baselines, while generalizing robustly across unseen network scales and channel conditions without retraining.
\end{abstract}


\section{Introduction}
Future wireless systems increasingly rely on dense spectrum reuse \cite{itur_m2160_2023,ericsson_mobility_report_2024} and multi-cell MU-MIMO transmission \cite{tataria2021sixg,letaief2019roadmap} to meet growing capacity and connectivity demands. However, aggressive co-channel reuse intensifies inter-cell and inter-user interference, making interference-aware resource coordination essential \cite{gesbert2010multicell,bjornson2013optimal}. In OFDMA downlink networks, discrete resource block group (RBG) scheduling determines the interference topology \cite{huang2009downlink}, whereas continuous beamformers control the desired signal power and interference leakage. Jointly optimizing these two variables therefore leads to a challenging nonlinear mixed-integer problem, in which each scheduling decision changes the interference environment over which the beamformers are optimized \cite{liu2014complexity,song2008weighted,yu2011multicell, hassan2014downlink,khanafer2012mimo}.

Classical model-driven methods typically address this coupling through alternating optimization. For a fixed RBG assignment, continuous beamforming can be optimized using WMMSE or fractional programming (FP) \cite{christensen2008wmmse,shi2011wmmse,shen2018fp1}. Although recent WMMSE variants reduce the antenna-dimensional computational burden \cite{zhao2023rethinking}, matrix-related operations and power-constraint handling can remain costly in large MIMO systems. FastFP provides an inversion-free alternative by constructing a quadratic surrogate from a conservative curvature bound, thereby ensuring monotonic ascent under the minorization--maximization principle \cite{sun2017majorization,shen2024accelerating}. For discrete scheduling, the transformed FP objective yields analytical edge utilities \cite{shen2018fp2}, allowing RBG allocation to be solved through weighted Hungarian matching \cite{khan2020fp_hungarian}. Nevertheless, this model-driven pipeline exhibits two complementary limitations: the continuous FastFP step can be overly conservative because it is governed by a worst-case curvature bound, whereas the discrete step repeatedly invokes expensive exact matching and myopically maximizes only the current surrogate, without accounting for its effect on subsequent interference and beamforming trajectories.

To address the above limitations, learning-based optimization has emerged as a promising approach for improving the
performance--complexity tradeoff of wireless resource allocation. Existing methods can be broadly organized along the two variable types considered in this work: learning-assisted continuous beamforming optimization and learning-based discrete resource allocation.

For continuous beamforming, early learning-based methods approximate the solution mapping of iterative optimizers or directly predict beamformers from channel states \cite{sun2018learning,lu2020learning}. Although such black-box
mappings can reduce online latency, they do not explicitly preserve the update structure, feasibility mechanism, or convergence behavior of model-driven solvers. Deep unfolding addresses this limitation by embedding a finite number of analytical iterations into a trainable architecture \cite{monga2021algorithm}. Existing unfolded beamforming methods are primarily based on WMMSE-type iterations: IAIDNN learns approximations to expensive matrix inverse operations \cite{hu2021iaidnn}, matrix-inverse-free WMMSE reformulates the updates to remove matrix inversion and related numerical procedures \cite{pellaco2022matrixfree}, and graph-based unfolding exploits interference topology to improve scalability \cite{chowdhury2024uwmmse}. Beyond WMMSE-based designs, DeepFP unfolds FastFP and demonstrates that its iterations can be improved through data-driven parameter tuning \cite{zhu2026deepfp}. Nevertheless, it remains unclear how to accelerate the conservative FastFP trajectory while explicitly preserving its analytical update direction and convergence guarantee. This motivates a bounded learning mechanism that introduces trainable flexibility only within a theoretically safe update region.

For discrete resource allocation, learning-based methods commonly adopt direct policy prediction or reinforcement learning \cite{liang2019deep,alwarafy2021drlsurvey}. However, directly predicting complete assignments faces combinatorial action spaces, difficult constraint satisfaction, and limited scalability. Structured learning-to-optimize methods instead embed trainable decisions into conventional optimization procedures \cite{shen2020lorm}, while topology-aware and permutation-equivariant architectures improve generalization across wireless network configurations \cite{shen2021gnn_rrm,wang2022decentralized_gnn}. Spatial deep learning has also been developed specifically for wireless scheduling by exploiting local interference relations \cite{cui2019spatial}. However, these approaches do not explicitly leverage the analytical FP edge utilities inherent to our matching subproblem. Since such utilities already quantify the local quality of each slot--RBG assignment, relearning the entire allocation mapping is unnecessary. A more structured alternative is to preserve these analytical weights and learn only the order in which competing logical slots select the available RBGs.

Motivated by these two gaps, we propose a joint deep-unfolding framework that integrates P-Net and K-Net into the FP-based alternating solver. Both modules follow a bounded learning principle: P-Net learns a relaxation factor constrained to an ascent-preserving interval along the analytical FastFP direction, whereas K-Net learns only the priority order used by an analytical greedy decoder. Consequently, the learned components enhance the continuous and discrete updates without replacing the reliable model-driven optimization structure \cite{liu2019alista,bertocchi2020irestnet,mukherjee2023guarantees}.


The main contributions of this paper are summarized as follows:
\begin{itemize}
    \item \textbf{Deep Unfolding Framework via Bounded Learning.}
    We formulate joint RBG scheduling and beamforming in multi-cell MU-MIMO networks as a mixed-integer WSR maximization problem. To address the bottlenecks of conservative continuous updates and myopic discrete matching in traditional alternating algorithms, we propose a joint deep-unfolding framework comprising P-Net and K-Net. Following a bounded learning principle, this architecture restricts the neural networks to structured, analytically controlled decision spaces, ensuring that learning enhances rather than replaces reliable model-driven optimization structures.

    \item \textbf{Accelerated Beamforming with Convergence Guarantees via P-Net.}
    For continuous beamforming optimization, we develop P-Net to learn an adaptive relaxation factor along the analytical FastFP direction and strictly bound it within a safe interval. Furthermore, we introduce a low-complexity surrogate safety-check and restart mechanism to handle Nesterov momentum. This bounded unfolding design allows P-Net to significantly accelerate the FastFP trajectory while retaining monotonically nondecreasing and stationary-point convergence.

    \item \textbf{Low-Complexity Scheduling via Priority-Based K-Net.}
    For the discrete matching component, we propose K-Net, which preserves the analytical FP edge utilities and learns a priority order for sequential greedy RBG decoding. Trained with terminal WSR feedback, K-Net accounts for the long-term impact of current scheduling decisions on subsequent trajectories. This structured approach avoids direct assignment prediction, inherently maintains allocation feasibility, and reduces the scheduling complexity from cubic to quadratic in the number of RBGs.

    \item \textbf{Superior Tradeoff and Strong Generalization.}
    Extensive simulations demonstrate that the joint P-Net~+~K-Net deep-unfolding framework achieves a superior WSR--runtime tradeoff over conventional model-driven baselines. The learned modules also generalize across unseen network conditions without retraining.
\end{itemize}

The remainder of this paper is organized as follows. Section~\ref{sec:system_model} presents the multi-cell MU-MIMO system model and formulates the joint RBG scheduling and beamforming problem. Section~\ref{sec:fp} reviews the FP-based beamforming update and Hungarian-based RBG allocation. Section~\ref{sec:unfolding} introduces the proposed bounded deep-unfolding framework, including P-Net and K-Net. Section~\ref{sec:exp_results} reports the numerical results, and Section~\ref{sec:conclusion} concludes the paper.

Throughout this paper, boldface uppercase and lowercase letters denote matrices and vectors, respectively. The space of $M \times N$ complex matrices is denoted by $\mathbb{C}^{M \times N}$. The operators $(\cdot)^T$ and $(\cdot)^H$ stand for the transpose and conjugate transpose of a matrix, respectively. The trace and determinant of a matrix are represented by $\operatorname{Tr}(\cdot)$ and $\det(\cdot)$, respectively. $\mathbf{I}_N$ denotes an $N \times N$ identity matrix. $\Re\{\cdot\}$ extracts the real part of a complex variable. $\mathbb{E}[\cdot]$ denotes the statistical expectation. $\|\cdot\|_F$ represents the Frobenius norm of a matrix. $\mathcal{P}_{\mathcal{W}}(\cdot)$ denotes the projection operator onto a given set $\mathcal{W}$. For Hermitian matrices $\mathbf{A}$ and $\mathbf{B}$, $\mathbf{A} \preceq \mathbf{B}$ indicates that $\mathbf{B} - \mathbf{A}$ is positive semi-definite, and $\lambda_{\max}(\mathbf{A})$ denotes the maximum eigenvalue of $\mathbf{A}$. The cardinality of a set $\mathcal{A}$ is denoted by $|\mathcal{A}|$, and $[N]$ defines the set of integers $\{0, 1, \dots, N-1\}$.

\section{System Model and Problem Formulation}
\label{sec:system_model}
\subsection{System Model}
We consider a multi-cell MU-MIMO downlink with BS set $\mathcal A$, user set $\mathcal U$, and shared RBG set $\mathcal K$, whose cardinalities are $B$, $U$, and $K$, respectively.  All BSs operate on a co-channel deployment basis, sharing the entire set of RBGs, which inherently induces inter-cell interference. Each BS and user are equipped with $M_t$ transmit and $M_r$ receive antennas, respectively. The serving BS transmits $d_u$ streams to user $u$, with data symbol $\mathbf s_u\in\mathbb C^{d_u}$ satisfying $\mathbb E[\mathbf s_u\mathbf s_u^H]=\mathbf I_{d_u}$.
\begin{figure*}[b!]
\normalsize
\hrulefill
\begin{equation}
\label{eq:f_P3_full}
\begin{aligned} 
f_{\text{P3}}(\mathbf{V},\mathbf{\Gamma},\mathbf{Y}) = \sum_{u,i} \Big[  \rho_u \log \det(\mathbf{I}_{d_u} + \mathbf{\Gamma}_{u,i}) - \rho_u \text{Tr}(\mathbf{\Gamma}_{u,i}) + \text{Tr}\left( 2\Re\{\mathbf{V}_{u,i}^H \mathbf{\Lambda}_{u,i}\} \right) 
& - \rho_u \text{Tr} \left( \mathbf{Y}_{u,i}^H \mathbf{J}_{u,i} \mathbf{Y}_{u,i} (\mathbf{I}_{d_u} + \mathbf{\Gamma}_{u,i}) \right) \Big].
\end{aligned}
\end{equation}
\end{figure*}
\subsection{Signal Model and Problem Formulation}
To streamline the formulation, we define two decision variables as follows:

\begin{itemize}
    \item \textbf{RBG allocation variable}
    $k_{u,i}\in[K]$:
    Each user $u$ has $N$ logical decision slots indexed by $i\in[N]$, with $N=K$ unless otherwise specified. This ensures that each user has at most one allocation decision for each physical RBG. As illustrated in Fig.~\ref{fig:rbg_desp}, $k_{u,i}$ maps the $i$-th logical decision slot of user $u$ to a physical RBG. If $k_{u,i}=\emptyset$, the slot is left unassigned; otherwise, the slot is assigned physical RBG $k_{u,i}$ and becomes active.
    \item \textbf{Transmit beamforming variable} 
    $\mathbf{V}_{b_u,k_{u,i},u}\in\mathbb{C}^{M_t\times d_u}$:
    For an active allocation decision slot $(u,i)$ with $k_{u,i}\neq\emptyset$, this precoding matrix is applied by the serving BS $b_u$ to transmit $d_u$ data streams to user $u$ on physical RBG $k_{u,i}$.
\end{itemize}

\begin{figure}[t]
    \centering
    \includegraphics[width=0.48\textwidth]{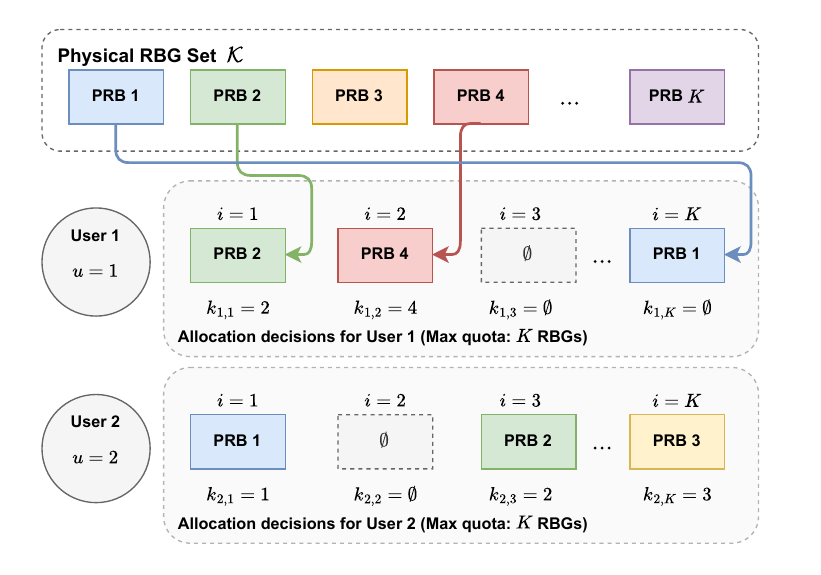}
    \caption{Illustration of the proposed resource allocation variables. Physical RBG indices (represented as blocks) are directly mapped into predefined logical decision slots for each user. Unused allocation decisions are simply marked with the $\emptyset$ state.}
    \label{fig:rbg_desp}
\end{figure}

Define the active slots sharing RBG $k_{u,i}$ as
\begin{equation}
\label{eq:C_ui_def}
\mathcal{C}_{u,i}
\triangleq
\{(j,\ell)\mid k_{j,\ell}=k_{u,i},\; k_{j,\ell}\neq\emptyset\}.
\end{equation}

Since $b_u$ and $k_{u,i}$ are determined by $(u,i)$, we subsequently abbreviate $\mathbf V_{b_u,k_{u,i},u}$ and $\mathbf H_{b_u,k_{u,i},u}$ as $\mathbf V_{u,i}$ and $H_{u,i}$, respectively. The corresponding interference-plus-noise covariance is given by:
\begin{equation}
\label{eq:M_matrix}
\mathbf{M}_{u,i}
=
\sum_{(j,\ell)\in\mathcal{C}_{-u,i}}
\mathbf{H}_{b_j,k_{u,i},u}
\mathbf{V}_{j,\ell}
\mathbf{V}_{j,\ell}^{H}
\mathbf{H}_{b_j,k_{u,i},u}^{H}
+
\sigma^2\mathbf{I}_{M_r}.
\end{equation}
where $\mathcal{C}_{-u,i} \triangleq \mathcal{C}_{u,i}\setminus\{(u,i)\}$.

The achievable rate $R_{u,i}$ for user $u$ on slot $i$ is directly formulated as:
\begin{equation}
R_{u,i} = \log \det \left(\mathbf{I}_{d_u} + \mathbf{V}_{u,i}^H \mathbf{H}_{u,i}^H \mathbf{M}_{u, i}^{-1} \mathbf{H}_{u,i} \mathbf{V}_{u,i}\right).
\end{equation}

Our objective is to maximize the network's Weighted Sum Rate (WSR), where $\rho_u \geq 0$ denotes the priority weight assigned to user $u$.
The joint resource allocation and beamforming optimization problem can be formulated as:

\begin{subequations}
\label{eq:P1}
\begin{align}
\text{P1:}\quad
\max_{\{k_{u,i}\},\,\{\mathbf{V}_{u,i}\}}
\quad & \sum_{u \in [U]} \rho_u \sum_{i \in [N]} R_{u,i} 
\label{eq:obj_wsr}\\
\text{s.t.}\quad
& \|\mathbf{V}_{u,i}\|_F^2 \leq P_{\max},
\quad \forall u \in [U],\,i \in [N],
\label{eq:power_constraint}\\
& k_{u,i} \in [K],
\quad \forall u \in [U],\,i \in [N].
\label{eq:rbg_constraint}
\end{align}
\end{subequations}

Problem \eqref{eq:P1} is a nonconvex mixed-integer program due to the discrete RBG assignments and interference-coupled log-determinant rates. We first derive an FP-based alternating solver in Section \ref{sec:fp}, and then embed learnable continuous and discrete updates into its iterations in Section \ref{sec:unfolding} to simultaneously accelerate convergence and enhance the achievable system performance.

\section{Fractional Programming based Alternating Optimization}
\label{sec:fp}

To solve \eqref{eq:P1}, we adopt an alternating optimization (AO) framework for the continuous beamforming variables $\mathbf{V}_{u,i}$ and discrete RBG allocations $k_{u,i}$. For fixed $k_{ui}$, FP sequentially decouples the desired signals and multi-user interference, transforming the log-determinant objective into tractable quadratic surrogate subproblems \cite{shen2018fp1}. With the FP auxiliary variables fixed, the transformed objective becomes additive over candidate slot--RBG assignments, reducing the allocation step to independent per-user weighted bipartite matching problems with analytically derived edge utilities, which are solved exactly by the Hungarian algorithm \cite{shen2018fp2}.

\subsection{Beamforming Optimization via the FP Framework}

To address the matrix ratio terms embedded within the log-determinant operation, we first apply the Lagrangian dual transform \cite{shen2018fp1}. By introducing an auxiliary variable $\mathbf{\Gamma}_{u,i} \in \mathbb{C}^{d_u \times d_u}$ and fixing $\mathbf{k}$, the original problem is equivalently recast as:
\begin{equation}
\label{eq:P2}
\begin{aligned} 
\text{P2:}\quad\max_{\mathbf{V},\mathbf{\Gamma}} \quad & \sum_{u, i} \Big[ \rho_u \log \det(\mathbf{I}_{d_u} + \mathbf{\Gamma}_{u,i}) - \rho_u \text{Tr}(\mathbf{\Gamma}_{u,i}) \\ 
& + \text{Tr} \left( \rho_u \mathbf{V}_{u,i}^H \mathbf{H}_{u,i}^H \mathbf{J}_{u,i}^{-1} \mathbf{H}_{u,i} \mathbf{V}_{u,i} (\mathbf{I}_{d_u} + \mathbf{\Gamma}_{u,i}) \right) \Big]
\end{aligned} 
\end{equation}
where $\mathbf{J}_{u,i} = \mathbf{H}_{u,i}\mathbf{V}_{u,i}\mathbf{V}_{u,i}^H\mathbf{H}_{u,i}^H + \mathbf{M}_{u,i}$ represents the total covariance matrix. Note that throughout problems P2--P4, the constraints in P1
are retained but omitted for brevity.

Although \eqref{eq:P2} successfully isolates the desired signal terms, the total covariance matrix $\mathbf{J}_{u,i}$ remains coupled within a matrix inverse. To fully decouple these variables, we apply the Quadratic Transform. By introducing decoding auxiliary variables $\mathbf{Y}_{u,i} \in \mathbb{C}^{M_r \times d_u}$, we establish a quadratic lower bound, yielding the following surrogate problem:
\begin{equation}
\label{eq:P3}
\max_{\mathbf{V},\mathbf{\Gamma},\mathbf{Y}} \quad f_{\text{P3}}(\mathbf{V},\mathbf{\Gamma},\mathbf{Y}),
\end{equation}
where the explicit expression of $f_{\text{P3}}$ is given by \eqref{eq:f_P3_full} at \textbf{the bottom of this page}, and 
\begin{equation}
\label{eq:Lambda_fp}
    \mathbf{\Lambda}_{u,i} \triangleq \rho_u \mathbf{H}_{u,i}^H \mathbf{Y}_{u,i} (\mathbf{I}_{d_u} + \mathbf{\Gamma}_{u,i}).
\end{equation}

We iteratively update $\mathbf{\Gamma}_{u,i}$, $\mathbf{Y}_{u,i}$, and $\mathbf{V}_{u,i}$ in a block coordinate ascent manner. Based on the first-order optimality conditions, the optimal closed-form updates for the auxiliary variables, when fixing the other variables, are respectively obtained as:
\begin{equation}
\label{eq:gamma_update}
\mathbf{\Gamma}_{u,i} = \mathbf{V}_{u,i}^H \mathbf{H}_{u,i}^H \mathbf{M}_{u, i}^{-1} \mathbf{H}_{u,i} \mathbf{V}_{u,i},
\end{equation}
\begin{equation}
\label{eq:y_update}
\mathbf{Y}_{u,i} = \mathbf{J}_{u,i}^{-1} \mathbf{H}_{u,i} \mathbf{V}_{u,i}.
\end{equation}
 
When optimizing the transmit beamformers $\mathbf{V}_{u,i}$ with the $\mathbf{\Gamma}_{u,i}$ and $\mathbf{Y}_{u,i}$ fixed, the surrogate function \eqref{eq:P3} becomes concave with respect to $\mathbf{V}_{u,i}$. Consequently, the optimal closed-form update is given by:
\begin{equation}
\label{eq:fp_update}
\mathbf{V}_{u,i}
=
\left(
\mathbf{L}_{u,i}
+
\mu_{u,i}\mathbf{I}_{M_t}
\right)^{-1}
\boldsymbol{\Lambda}_{u,i},
\end{equation}
where $\mu_{u,i}\geq 0$ is the Lagrange multiplier associated with the power constraint. To express the curvature matrix $\mathbf{L}_{u,i}$ compactly, we define the weighted decoding matrix as:
\begin{equation}
\label{eq:Psi_def}
\boldsymbol{\Psi}_{j,\ell}
\triangleq
\mathbf{Y}_{j,\ell}
\left(
\mathbf{I}_{d_j}
+
\boldsymbol{\Gamma}_{j,\ell}
\right)
\mathbf{Y}_{j,\ell}^{H}.
\end{equation}
Thus, the curvature matrix $\mathbf{L}_{u,i}$ is formulated as:
\begin{equation}
\label{eq:L_matrix_fp}
\mathbf{L}_{u,i}
=
\sum_{(j,\ell)\in\mathcal{C}_{u,i}}
\rho_j
\mathbf{H}_{b_u,k_{ui},j}^{H}
\boldsymbol{\Psi}_{j,\ell}
\mathbf{H}_{b_u,k_{ui},j}.
\end{equation}

\subsection{FastFP and Nesterov Acceleration}
\begin{figure*}[b!]
\normalsize
\hrulefill
\begin{equation}
\label{eq:f_P4_full}
\begin{aligned} 
f_{\text{P4}}(\mathbf{V},\mathbf{\Gamma},\mathbf{Y},\mathbf{Z}) = \sum_{u,i} \Big[ & \rho_u \log \det(\mathbf{I}_{d_u} + \mathbf{\Gamma}_{u,i}) - \rho_u \text{Tr}(\mathbf{\Gamma}_{u,i}) - \rho_u \sigma^2 \text{Tr}\left(\mathbf{Y}_{u,i}^H \mathbf{Y}_{u,i} (\mathbf{I}_{d_u} + \mathbf{\Gamma}_{u,i})\right) - \lambda_{u,i} \text{Tr}(\mathbf{V}_{u,i}^H \mathbf{V}_{u,i}) \\
& + \text{Tr}\left( 2\Re\left\{ \mathbf{V}_{u,i}^H \mathbf{\Lambda}_{u,i} + \mathbf{V}_{u,i}^H (\lambda_{u,i} \mathbf{I}_{M_t} - \mathbf{L}_{u,i}) \mathbf{Z}_{u,i} \right\} \right) - \text{Tr}\left( \mathbf{Z}_{u,i}^H (\lambda_{u,i} \mathbf{I}_{M_t} - \mathbf{L}_{u,i}) \mathbf{Z}_{u,i} \right) \Big]. 
\end{aligned}
\end{equation}
\end{figure*}

Although the FP beamforming update in \eqref{eq:fp_update} admits a closed-form solution, it requires an $M_t$-dimensional matrix inversion with complexity $\mathcal{O}(M_t^3)$. FastFP avoids this operation by applying the nonhomogeneous bound \cite{sun2017majorization,shen2024accelerating}. With the FP auxiliary variables $\mathbf{\Gamma}$ and $\mathbf{Y}$ fixed at their closed-form updates in \eqref{eq:gamma_update} and \eqref{eq:y_update}, respectively, we choose
\begin{equation}
    \lambda_{u,i}
    \geq
    \lambda_{\max}\!\left(\mathbf{L}_{u,i}\right),
\end{equation}
such that
$\lambda_{u,i}\mathbf{I}_{M_t}-\mathbf{L}_{u,i}\succeq\mathbf{0}$. By introducing the auxiliary variables
$\mathbf{Z}$ and applying the nonhomogeneous bound to the beamformer-dependent quadratic terms, we obtain a tight minorizing surrogate of \eqref{eq:P3}:
\begin{equation}
\label{eq:P4}
\max_{\mathbf{V},\mathbf{\Gamma},\mathbf{Y},\mathbf{Z}} \quad f_{\text{P4}}(\mathbf{V},\mathbf{\Gamma},\mathbf{Y},\mathbf{Z}),
\end{equation}
where the explicit expression of $f_{\text{P4}}$ is given by \eqref{eq:f_P4_full} at \textbf{the bottom of this page}.



At iteration $\tau$, maximizing the surrogate in \eqref{eq:P4} over the feasible power set yields the inversion-free projected update as follows:
\begin{equation}
\label{eq:fastfp_projected_update}
\begin{aligned}
\mathbf{V}_{u,i}^{(\tau)}
=
\mathcal{P}_{\mathcal{W}}
\Bigg[
\mathbf{Z}_{u,i}^{(\tau)}
+
\frac{1}{\lambda_{u,i}}
\Big(
\mathbf{\Lambda}_{u,i}
-
\mathbf{L}_{u,i}\mathbf{Z}_{u,i}^{(\tau)}
\Big)
\Bigg],
\end{aligned}
\end{equation}
where the projection operation $\mathcal{P}_{\mathcal{W}}(\cdot)$ onto the per-beamformer power constraint is simply given by
$\mathcal{P}_{\mathcal{W}}(\mathbf{X})=\min\left\{1,\sqrt{P_{\max}}/\|\mathbf{X}\|_F\right\}\mathbf{X}$.

The specific FastFP variant is determined entirely by the choice of the anchor matrix $\mathbf{Z}_{u,i}^{(\tau)}$. Standard FastFP uses the previous iterate as the anchor, i.e.,
$\mathbf{Z}_{u,i}^{(\tau)}=\mathbf{V}_{u,i}^{(\tau-1)}$,
under which \eqref{eq:fastfp_projected_update} becomes a projected gradient-ascent step. In contrast, Nesterov-FastFP replaces the previous-iterate anchor with a momentum-extrapolated point \cite{shen2024accelerating}: $\mathbf Z_{u,i}^{(\tau)}
=\mathbf U_{u,i}^{(\tau-1)}$, where
\begin{equation}
\label{eq:nesterov_anchor_update}
\begin{aligned}
\mathbf{U}_{u,i}^{(\tau-1)}
={}&
\mathbf{V}_{u,i}^{(\tau-1)}
+
\nu^{(\tau-1)}
\Big(
\mathbf{V}_{u,i}^{(\tau-1)}
-
\mathbf{V}_{u,i}^{(\tau-2)}
\Big),\\
\nu^{(\tau-1)}
={}&
\max\left\{
\frac{\tau-2}{\tau+1},\,0
\right\}.
\end{aligned}
\end{equation}
Therefore, Standard FastFP and Nesterov-FastFP share the same projected update in \eqref{eq:fastfp_projected_update} and differ only in the construction of the anchor matrix.

\subsection{RBG Allocation Optimization via Bipartite Matching}


For fixed FP auxiliary variables, the surrogate in \eqref{eq:P4} is additive over candidate slot--RBG assignments. Thus, the global allocation can be decoupled by evaluating the potential WSR gain of assigning physical RBG $n$ to user $u$'s $i$-th decision slot. We define this gain as the edge weight $\xi_{u,i,n}^{(\tau)}$. To formulate the edge weight explicitly, we first define the set of active allocation decisions from other users currently occupying RBG $n$ as:
\begin{equation}
\mathcal{S}_{-u,n}^{(\tau-1)}
\triangleq
\{(j,\ell)\mid j\neq u,\; k_{j,\ell}^{(\tau-1)}=n\}.
\end{equation}
Consequently, the receiver-side allocation decisions affected by the candidate beamformer are collected in
$\mathcal{R}_{u,i,n}^{(\tau)} \triangleq \{(u,i)\} \cup \mathcal{S}_{-u,n}^{(\tau-1)}.$
The edge weight $\xi_{u,i,n}^{(\tau)}$ is then formulated as:
\begin{equation}
\begin{aligned}
\label{eq:xi_weight}
\xi&_{u,i,n}^{(\tau)} =  \, 2 \Re\left\{\operatorname{Tr}\left(\widehat{\mathbf{V}}_{u,i,n}^{(\tau)H} \mathbf{\Lambda}_{u,i,n}^{(\tau)}\right)\right\} - \operatorname{Tr}\left(\widehat{\mathbf{V}}_{u,i,n}^{(\tau)H} \mathbf{L}_{u,i,n}^{(\tau)} \widehat{\mathbf{V}}_{u,i,n}^{(\tau)} \right) \\
 & - \sum_{(j,\ell)\in\mathcal S_{-u,n}^{(\tau-1)}} \rho_u \operatorname{Tr}\left( \mathbf{V}_{j,\ell}^{(\tau-1)H} \mathbf{H}_{b_j,n,u}^{H} \boldsymbol{\Psi}_{u,i}^{(\tau)} \mathbf{H}_{b_j,n,u} \mathbf{V}_{j,\ell}^{(\tau-1)} \right),
\end{aligned}
\end{equation}
where the optimal candidate beamformer $\widehat{\mathbf{V}}_{u,i,n}^{(\tau)}$ is given by:
\begin{equation}
\label{eq:v_cand}
\widehat{\mathbf{V}}_{u,i,n}^{(\tau)} = \mathcal{P}_{\mathcal{W}}\left( \mathbf{V}_{u,i}^{(\tau-1)} + \frac{1}{\lambda_{u,i,n}^{(\tau)}} \left(\mathbf{\Lambda}_{u,i,n}^{(\tau)} - \mathbf{L}_{u,i,n}^{(\tau)} \mathbf{V}_{u,i}^{(\tau-1)} \right) \right).
\end{equation} 
The constituent matrices $\boldsymbol{\Lambda}_{u,i,n}^{(\tau)}$ and $\mathbf{L}_{u,i,n}^{(\tau)}$ are defined as:
\begin{subequations}
\begin{align}
\boldsymbol{\Lambda}_{u,i,n}^{(\tau)}
  &= \rho_u \,
     \mathbf{H}_{b_u,n,u}^{H}
     \mathbf{Y}_{u,i}^{(\tau)}
     \bigl(\mathbf{I}_{d_u} + \boldsymbol{\Gamma}_{u,i}^{(\tau)}\bigr), \\
\mathbf{L}_{u,i,n}^{(\tau)}
  &= \sum_{(j,\ell) \in \mathcal{R}_{u,i,n}^{(\tau)}}
     \rho_j \,
     \mathbf{H}_{b_u,n,j}^{H}
     \boldsymbol{\Psi}_{j,\ell}^{(\tau)}
     \mathbf{H}_{b_u,n,j}.
\end{align}
\end{subequations}
As shown in \eqref{eq:xi_weight}, this edge weight captures three critical components: the desired signal contribution of the candidate link, the interference leakage to the receiver-side links in $\mathcal{R}_{u,i,n}^{(\tau)}$, and the interference received from existing transmitters on RBG $n$.

With the edge weights analytically established, the combinatorial allocation problem becomes a \textbf{maximum-weight bipartite matching problem}. Since our system allows spatial sharing of physical RBGs, the global allocation decouples into $U$ independent per-user matching subproblems. For each user $u$, a bipartite graph connects its $N$ logical decision slots to the $K$ physical RBGs. The optimal discrete assignment $x_{u,i,n}^* \in \{0,1\}$ is obtained by solving the following integer linear program via the Hungarian algorithm:
\begin{subequations}
\label{eq:bipartite}
\begin{align}
\underset{x_{u,i,n} \in \{0,1\}}{\text{maximize}} \quad & \sum_{i \in [N]} \sum_{n \in [K]} \xi_{u,i,n}^{(\tau)} x_{u,i,n} \\
\text{s.t.} \quad & \sum_{n \in [K]} x_{u,i,n} \leq 1, \quad \forall i \in [N] \\
& \sum_{i \in [N]} x_{u,i,n} \leq 1, \quad \forall n \in [K]
\end{align}
\end{subequations}
Upon obtaining the optimal binary assignment $x_{u,i,n}^*$, the specific RBG index $k_{ui}$ allocated to user $u$'s $i$-th decision slot is recovered by:
\begin{equation}
\label{eq:recoverk_from_x}
k_{ui} = 
\begin{cases} 
n, & \text{if } x_{u,i,n}^* = 1 \text{ and }  \xi^{(\tau)}_{u,i,n}>0,\\ 
\emptyset, & \text{otherwise}.
\end{cases}
\end{equation}

Note that in \eqref{eq:recoverk_from_x}, a candidate assignment is accepted only if it yields a strictly positive WSR gain, i.e., $\xi_{u,i,n}^{(\tau)}>0$. Otherwise, the slot is left unassigned, $k_{u,i}=\emptyset$, to prevent an assignment that would cause severe performance degradation through additional interference. 

The overall alternating procedure for joint RBG allocation and continuous beamforming is summarized in Algorithm~\ref{alg:rbg_matching}. The resulting AO solver has two main bottlenecks. First, the worst-case curvature bound $\lambda_{u,i}$ leads to conservative beamforming steps. Second, per-user Hungarian matching incurs a cubic assignment complexity of $\mathcal O(UK^3)$. These limitations motivate the continuous P-Net accelerator and the discrete K-Net scheduler developed next.

\begin{algorithm}[t!]
\caption{Joint Beamforming Optimization and RBG Allocation via Bipartite Matching}
\label{alg:rbg_matching}
\begin{algorithmic}[1]
\STATE \textbf{input:} Channel matrices $\mathbf{H}$.
\STATE \textbf{initialize:} Beamforming matrices $\mathbf{V}^{(0)}$ satisfying power constraints, and RBG allocation $\mathbf{k}^{(0)}$.
\REPEAT 
    \STATE \textbf{[Beamforming Optimization Step]} \hfill \textit{At iteration $\tau$}
    \STATE Update $\mathbf{\Gamma}^{(\tau)}$ and $\mathbf{Y}^{(\tau)}$ via \eqref{eq:gamma_update} and \eqref{eq:y_update}.

    \STATE Update $\mathbf{V}^{(\tau)}$ via \eqref{eq:fastfp_projected_update}.
    \STATE \textbf{[RBG Allocation Optimization Step]}

    \FOR{each user $u \in [U]$}
        \STATE Compute $\mathbf{\Xi}_u^{(\tau)} = [\xi_{u,i,n}^{(\tau)}]_{N \times K}$ via \eqref{eq:xi_weight}, $\forall i \in [N], n \in [K]$.
        \STATE Solve \eqref{eq:bipartite} via Hungarian algorithm to obtain binary assignment $x_{u,i,n}^*$.
        \STATE Recover $k_{ui}^{(\tau)}$ from $x_{u,i,n}^*$ via \eqref{eq:recoverk_from_x} and update $\mathbf{V}_{u,i}^{(\tau)} \leftarrow \widehat{\mathbf{V}}_{u,i,k_{ui}^{(\tau)}}^{(\tau)} ~,\forall i \in [N]$.
    \ENDFOR
\UNTIL{the objective WSR \eqref{eq:P1} converges}
\STATE \textbf{output:} Final RBG allocation $\mathbf{k}^*$ and beamforming $\mathbf{V}^*$.
\end{algorithmic}
\end{algorithm}

\section{Bounded Deep Unfolding Framework}
\label{sec:unfolding}

We propose a bounded deep unfolding framework comprising P-Net and K-Net. P-Net learns a bounded relaxation factor to accelerate the continuous FastFP updates, whereas K-Net learns the slot priority of a sequential greedy decoder to replace repeated Hungarian matching. This strategically reduces the exact matching complexity from $\mathcal{O}(UK^3)$ to $\mathcal{O}(UK^2)$, allowing the network to directly optimize the long-term allocation trajectory. Both modules retain the analytical FP structure while optimizing the finite-horizon trajectory, as illustrated in Fig.~\ref{fig:total_net_arch}. 

\begin{figure*}[!t]
    \centering
    \includegraphics[width=1\textwidth]{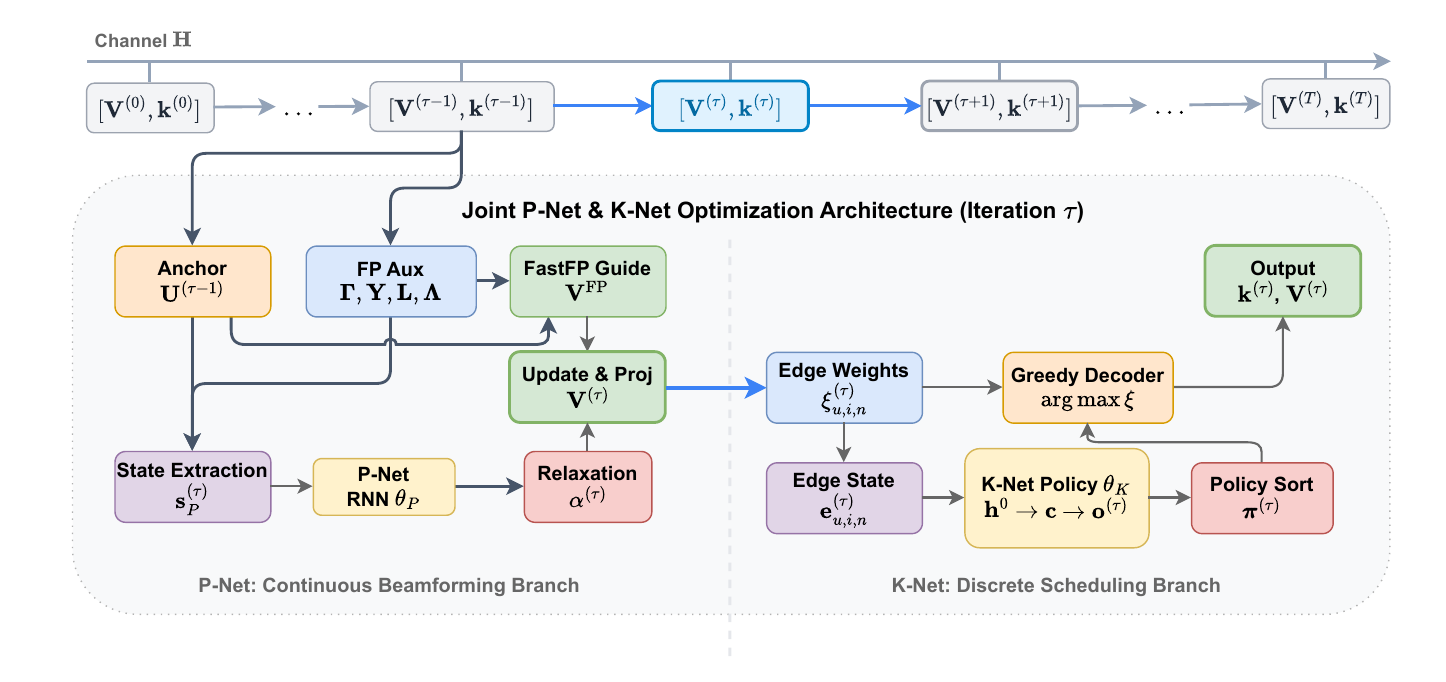}
    \caption{Joint P-Net and K-Net architecture. At iteration $\tau$, P-Net learns a bounded relaxation factor $\alpha^{(\tau)}$ along the FastFP direction, after which K-Net predicts slot priorities from the analytical edge utilities $\xi_{u,i,n}^{(\tau)}$ to guide greedy RBG allocation.}
    \label{fig:total_net_arch}
\end{figure*}

\subsection{P-Net: Bounded Relaxation Controller for Beamforming}

\subsubsection{Learned Update Rule and Its Geometric Principle}
Conventional FastFP is locally surrogate-driven: at each iteration, it moves directly to the constrained maximizer of the current quadratic minorant. Although this update provides a reliable ascent direction, repeatedly taking the same unit step does not exploit information accumulated along the optimization trajectory. We therefore retain the analytical FastFP direction and learn only how far to move along it.

Specifically, we set
$\mathbf Z_{u,i}^{(\tau)}=\mathbf V_{u,i}^{(\tau-1)}$.
With the FP auxiliary variables fixed, the beamformer-dependent part of the FastFP minorant can be written as
\begin{equation}
\label{eq:paraboloid}
\begin{aligned}
Q_{\tau}\!\left(
\mathbf V\mid\mathbf V^{(\tau-1)}
\right)&
={}
\sum_{u,i}
\Big[
2\Re\!\left\{
\operatorname{Tr}\!\left(
\mathbf V_{u,i}^{H}
\mathbf\Xi_{u,i}^{(\tau)}
\right)
\right\}\\
&\hspace{17mm}
-\lambda_{u,i}^{(\tau)}
\|\mathbf V_{u,i}\|_{F}^{2}
\Big]
+C_{\tau},
\end{aligned}
\end{equation}
where
$\mathbf\Xi_{u,i}^{(\tau)}
=
\mathbf\Lambda_{u,i}^{(\tau)}
+
\bigl(
\lambda_{u,i}^{(\tau)}\mathbf I_{M_t}
-\mathbf L_{u,i}^{(\tau)}
\bigr)
\mathbf V_{u,i}^{(\tau-1)}$,
and $C_{\tau}$ collects the terms independent of $\mathbf V$.

Equation~\eqref{eq:paraboloid} is a sum of blockwise isotropic concave paraboloids, each centered at its unconstrained maximizer. We denote this center by
$\widetilde{\mathbf V}_{u,i}^{\mathrm{FP},(\tau)}$, which is precisely the FastFP update in \eqref{eq:fastfp_projected_update} obtained by omitting the projection.

The symmetry of each paraboloid provides the geometric motivation for introducing a relaxation factor. Starting from
$\mathbf V_{u,i}^{(\tau-1)}$, moving toward the paraboloid center increases the surrogate, while moving beyond the center remains ascent-preserving as long as the step does not pass the reflection of the current point about that center. As illustrated in Fig.~\ref{fig:monotonicity_geom}, this corresponds exactly to the interval
$\alpha_{u,i}^{(\tau)}\in(0,2)$.

\begin{figure}[htbp]
    \centering
    \includegraphics[width=0.48\textwidth]{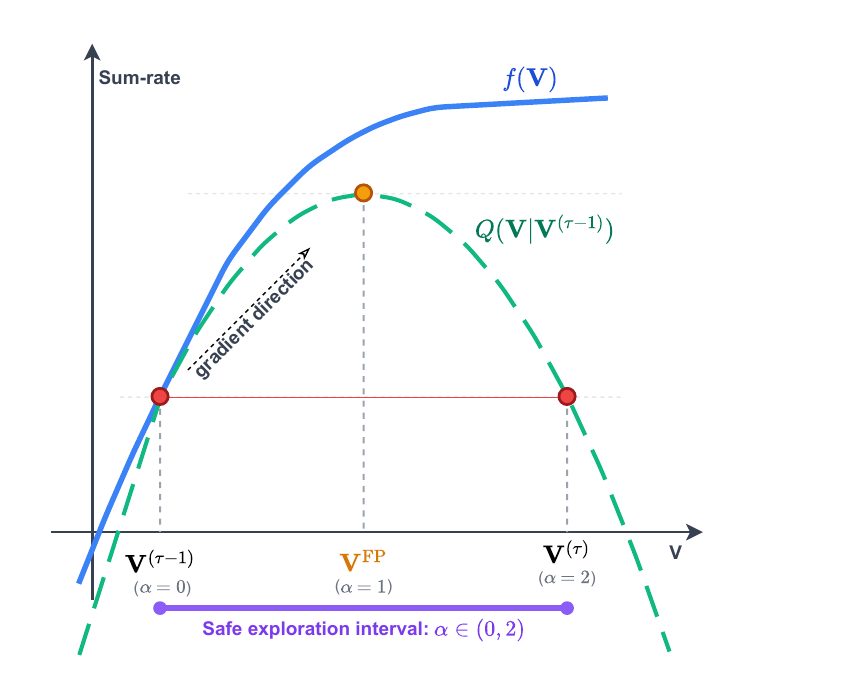}
    \caption{Geometric illustration of the monotonically nondecreasing guarantee. Since the FastFP minorant $Q$ is symmetric around its maximizer $\mathbf V^{\rm FP}$, any $\alpha\in(0,2)$ yields a surrogate value no lower than that at the anchor.}
    \label{fig:monotonicity_geom}
\end{figure}

Motivated by this property, P-Net generates a link-specific relaxation factor and performs the following projected update:
\begin{equation}
\label{eq:update_rule_pnet0}
\mathbf V_{u,i}^{(\tau)}
=
\mathcal P_{\mathcal W_{u,i}}
\left[
\mathbf V_{u,i}^{(\tau-1)}
+
\alpha_{u,i}^{(\tau)}
\left(
\widetilde{\mathbf V}_{u,i}^{\mathrm{FP},(\tau)}
-\mathbf V_{u,i}^{(\tau-1)}
\right)
\right].
\end{equation}
When $\alpha_{u,i}^{(\tau)}=1$, \eqref{eq:update_rule_pnet0} reduces to the conventional projected FastFP update. Values in $(0,1)$ yield under-relaxation, whereas values in $(1,2)$ provide controlled over-relaxation beyond the surrogate maximizer. Thus, P-Net adjusts the progress along the analytically derived FastFP direction. The independent per-beamformer projection preserves the monotonic nondecreasing property suggested by this geometry, as formally established in the following section.

\subsubsection{Monotonicity and Convergence Guarantees}
We next show that the same interval preserves ascent after independently projecting each beamformer onto its feasible set. For a fixed RBG allocation $\mathbf{k}$, let
$f(\mathbf V;\mathbf k)$ denote the WSR objective in \eqref{eq:P1}. 

\begin{lemma}[Monotonicity of the bounded P-Net update]
\label{lemma:monotonicity}
Fix an RBG allocation $\mathbf{k}$.
At iteration $\tau$, let the FP auxiliary variables be updated
at $\mathbf{V}^{(\tau-1)}$, and choose
$\lambda_{u,i}^{(\tau)}
\geq
\lambda_{\max}(\mathbf{L}_{u,i}^{(\tau)})$
for every active slot $(u,i)$.
If
$\alpha_{u,i}^{(\tau)}\in(0,2)$
for all active $(u,i)$, then the bounded update in \eqref{eq:update_rule_pnet0} after projection
satisfies
\begin{equation}
f\!\left(\mathbf{V}^{(\tau)};\mathbf{k}\right)
\geq
f\!\left(\mathbf{V}^{(\tau-1)};\mathbf{k}\right).
\end{equation}
Consequently, the WSR sequence generated by the beamforming
updates is monotonically nondecreasing for fixed
$\mathbf{k}$.
\end{lemma}

The proof is given in Appendix~\ref{app:proof_lemma1}. Building upon Lemma 1, the following theorem establishes that every accumulation point of the generated beamformer sequence is stationary.

\begin{theorem}
\label{theorem:stationary}
\label{theorem}
Assume that the relaxation factors remain in a compact subset of $(0,2)$. For an RBG allocation $\mathbf k$, every accumulation point of the beamformer sequence generated by \eqref{eq:update_rule_pnet0} is stationary for the beamforming subproblem of \eqref{eq:P4} and the WSR objective sequence converges.
\end{theorem}

The proof is given in Appendix~\ref{app:proof_theorem1}. By strictly bounding the learnable relaxation factor, the proposed P-net successfully decouples the pursuit of acceleration from the risk of divergence.

\subsubsection{P-Net with Nesterov Acceleration}

P-Net incorporates Nesterov acceleration by replacing the previous iterate in the FastFP guide with the extrapolated anchor $\mathbf U^{(\tau-1)}$ defined in \eqref{eq:nesterov_anchor_update}. Since extrapolation may destroy monotonicity, the FP auxiliary variables $\mathbf\Gamma^{(\tau)}$ and $\mathbf Y^{(\tau)}$ are first evaluated at $\mathbf V^{(\tau-1)}$ and kept fixed during the following safety check. The momentum candidate is accepted only if
\begin{equation}
\label{eq:nesterov_safeguard}
f_{\rm P3}\!\left(
\mathbf V_{\rm cand}^{(\tau)},
\mathbf\Gamma^{(\tau)},
\mathbf Y^{(\tau)}
\right)
\ge
f_{\rm P3}\!\left(
\mathbf V^{(\tau-1)},
\mathbf\Gamma^{(\tau)},
\mathbf Y^{(\tau)}
\right).
\end{equation}
Because the transformed objective is a lower bound on the WSR and is tight at $\mathbf V^{(\tau-1)}$, this condition guarantees a nondecreasing WSR. If the check fails, we set $\mathbf U^{(\tau-1)}=\mathbf V^{(\tau-1)}$ and the update is recomputed from the anchor. Since the check requires only the existing surrogate quantities, it adds limited overhead while preserving the guarantee of monotonic nondecrease.

\subsubsection{Network Architecture and Training}
To avoid excessive physical features and dimensionality explosion, we construct the input state of P-Net using an optimization-state-based feature vector. For any active link serving user $u$ on slot $i$ at iteration $\tau$, we define the FastFP update direction and the Nesterov momentum direction as $\mathbf{D}_{u,i}^{(\tau)} = \mathbf{V}_{u,i}^{\text{FP}} - \mathbf{U}_{u,i}^{(\tau-1)}$ and $\mathbf{M}_{u,i}^{(\tau)} = \mathbf{U}_{u,i}^{(\tau-1)} - \mathbf{V}_{u,i}^{(\tau-1)}$, respectively. The six-dimensional compact input feature vector $\mathbf{s}_{u,i}^{(\tau)}$ is constructed entirely from quantities already available during the standard FastFP update:
\begin{equation}
\label{eq:s_feat}
\mathbf{s}_{u,i}^{(\tau)} = \big[ s_{1}, s_{2}, s_{3}, s_{4}, s_{5}, s_{6} \big].
\end{equation}
These six features sequentially correspond to:
\begin{itemize}
    \item \textit{Update magnitude} ($s_{1} = \frac{\|\mathbf{D}_{u,i}^{(\tau)}\|_F}{\|\mathbf{U}_{u,i}^{(\tau-1)}\|_F}$): captures the relative step size of the current FastFP guide.
    \item \textit{Curvature tightness} ($s_{2} = \frac{\text{Tr}(\mathbf{L}_{u,i})/M_t}{\lambda_{u,i}}$): evaluates whether the worst-case majorization bound is overly conservative.
    \item \textit{Power occupation} ($s_{3} = \frac{\|\mathbf{U}_{u,i}^{(\tau-1)}\|_F^2}{P_{\max}}$): indicates the proximity to the feasible power boundary.
    \item \textit{Projection risk} ($s_{4} = \max\Big(0, \frac{\|\mathbf{V}_{u,i}^{\text{FP}}\|_F^2}{P_{\max}} - 1\Big)$): assesses the likelihood of the candidate being truncated by the power projection.
    \item \textit{Momentum consistency} ($s_{5} = \frac{\Re\{\text{Tr}((\mathbf{D}_{u,i}^{(\tau)})^H \mathbf{M}_{u,i}^{(\tau)})\}}{\|\mathbf{D}_{u,i}^{(\tau)}\|_F \|\mathbf{M}_{u,i}^{(\tau)}\|_F }$): measures the alignment between the Nesterov momentum and the current gradient direction to determine momentum reliability.
    \item \textit{Iteration stage} ($s_{6} = \frac{\tau}{\tau + c}$): allows the network to distinguish between early-stage exploration and late-stage refinement via a bounded scale constant $c>0$.
\end{itemize}
These features summarize the FastFP step, local curvature, constraint activity, momentum consistency, and iteration progress. 

A recurrent controller $\mathcal F_{\theta}$ shared across all unfolding iterations maps these features to a bounded relaxation factor:
\begin{equation}
\begin{aligned}
\bigl(z_{u,i}^{(\tau)},\mathbf h_{u,i}^{(\tau)}\bigr)
&=
\mathcal F_{\theta}\!\left(
\mathbf s_{u,i}^{(\tau)},
\mathbf h_{u,i}^{(\tau-1)}
\right),\\
\alpha_{u,i}^{(\tau)}
&=
2\,\operatorname{Sigmoid}\!\left(z_{u,i}^{(\tau)}\right).
\end{aligned}
\end{equation}
The resulting $\alpha_{u,i}^{(\tau)}\in(0,2)$ is used in
\eqref{eq:update_rule_pnet0}. Parameter sharing keeps the model size independent of the unfolding depth and permits inference beyond the training horizon.

P-Net is trained without beamforming labels by minimizing the negative terminal WSR:
\begin{equation}
\label{eq:pnet_loss}
\mathcal L_{\rm P}(\theta)
=
-\mathbb E_{(\mathbf H,\mathbf V^{(0)},\mathbf k)\sim\mathcal D}
\left[
\sum_{u\in[U]}\rho_u
\sum_{i\in[N]}R_{u,i}^{(T)}(\theta)
\right].
\end{equation}
We employ truncated backpropagation through time while propagating the optimization variables and recurrent states over the complete rollout. The resulting training procedure is summarized in Algorithm~\ref{alg:bounded_training}.
\subsection{K-Net: Policy-Guided Priority RBG Scheduling}
Hungarian matching maximizes the current FP surrogate but does not account for how the selected RBG pattern affects subsequent interference and beamforming updates. K-Net therefore learns only the priority order in which logical slots access the analytical edge utilities $\xi_{u,i,n}^{(\tau)}$. The ordered slots then select RBGs through a sequential greedy decoder, reducing the assignment complexity from $\mathcal O(UK^3)$ to $\mathcal O(UK^2)$.

\subsubsection{Policy Formulation and Input Features}
\label{subsubsec:knet_mdp_framework}
K-Net is formulated as a finite-horizon policy embedded in the alternating optimization loop. At iteration $\tau$, the input state aggregates the underlying multi-cell topology and algorithmic states, which will be detailed later.
Its action
$a^{(\tau)}
=\{\boldsymbol{\pi}_{u}^{(\tau)}\}_{u=1}^{U}$
contains one slot-priority permutation for each user. The priority-greedy decoder and the subsequent analytical updates deterministically produce the next state. Conditioned on $\mathsf{s}^{(\tau)}$, the per-user permutations are sampled independently, giving the joint policy
\begin{equation}
\label{eq:knet_joint_policy}
p_{\theta}\!\left(a^{(\tau)}\mid\mathsf{s}^{(\tau)}\right)
=
\prod_{u=1}^{U}
p_{\theta}\!\left(
\boldsymbol{\pi}_{u}^{(\tau)}
\mid\mathsf{s}^{(\tau)}
\right).
\end{equation}
For each training sample $b$, the reference trajectory starts from the same initialization and applies FastFP with Hungarian matching for the same horizon $T$. The reward is the terminal WSR gain and no temporal discounting is required.

For the state feature, we construct a compact feature vector $\mathbf{e}_{u,i,n}^{(\tau)}$ for each candidate edge $(u,i,n)$, which contains the following features:
\begin{itemize}
    \item \textit{Edge Utility:}
    The analytical weight $\xi_{u,i,n}^{(\tau)}$, together with its
    desired-signal and interference-penalty components in
    \eqref{eq:xi_weight}, describes the candidate assignment utility.

    \item \textit{Allocation Status:}
    The binary assignment $x_{u,i,n}^{(\tau-1)}$ and slot-activation flag
    $\delta_{u,i}^{(\tau-1)}
    =\sum_{m\in[K]}x_{u,i,m}^{(\tau-1)}$
    indicate the current allocation state.

    \item \textit{Candidate Power:}
    The normalized candidate power
    $p_{u,i,n}^{(\tau)}
    =\|\widehat{\mathbf V}_{u,i,n}^{(\tau)}\|_F^2/P_{\max}$
    measures proximity to the power boundary.

    \item \textit{Effective Rate:}
    The quantity
    $r_{u,i}^{\rm eff,(\tau)}
    =d_u^{-1}\log\det(
    \mathbf I_{d_u}+\boldsymbol{\Gamma}_{u,i}^{(\tau)})$
    summarizes the current link quality.

    \item \textit{Interference Sensitivity:}
    The normalized decoding trace
    $\chi_{u,i}^{(\tau)}
    =\operatorname{Tr}(
    \boldsymbol{\Psi}_{u,i}^{(\tau)})/(M_r d_u)$
    reflects sensitivity to co-channel interference.

    \item \textit{Iteration Progress:}
    The bounded index
    $\eta_\tau=\tau/(\tau+c_{\rm K})$
    distinguishes early and late rollout stages.
\end{itemize}

\subsubsection{Global-Pooling Priority Network}
K-Net maps the feature vector
$\mathbf{E}^{(\tau)}=\{\mathbf{e}_{u,i,n}^{(\tau)}\}$
to slot-priority logits through shared edge encoding and hierarchical
mean pooling. Each candidate edge is first embedded as
$\mathbf{h}_{u,i,n}^{0}
=\psi_e(\mathbf{e}_{u,i,n}^{(\tau)})\in\mathbb{R}^{d_h}$.
Let $\mathcal{Q}=\{S,U,R,B,G\}$ denote the slot-, user-, RBG-,
BS-, and global-level contexts, respectively. Their contextual
representations are compactly written as
\begin{equation}
\mathbf{c}_{u,i,n}^{q}
=
\psi_q\left(
\frac{1}{|\Omega_{u,i,n}^{q}|}
\sum_{(u',i',n')\in\Omega_{u,i,n}^{q}}
\mathbf{h}_{u',i',n'}^{0}
\right),
\qquad q\in\mathcal{Q},
\label{eq:knet_hierarchical_pooling}
\end{equation}
where $\Omega_{u,i,n}^{q}$ respectively contains the edges associated
with slot $(u,i)$, user $u$, RBG $n$, serving BS $b_u$, or the entire
network.

The resulting contexts are fused with the local edge embedding, and
the priority logit of slot $(u,i)$ is obtained by averaging over its
candidate RBGs:
\begin{equation}
o_{u,i}^{(\tau)}
=
\psi_{\mathrm{out}}\left(
\frac{1}{K}\sum_{n\in[K]}
\psi_F\left(
\mathbf{h}_{u,i,n}^{0}
+
\sum_{q\in\mathcal{Q}}\mathbf{c}_{u,i,n}^{q}
\right)
\right).
\label{eq:knet_slot_logit}
\end{equation}
All network modules are shared across users, slots, RBGs, BSs, and
unfolding iterations. Therefore, the number of trainable parameters is
independent of $(U,N,K)$, while the forward complexity scales linearly
with the number of candidate edges.

\subsubsection{Priority-Greedy RBG Decoding}
During training, K-Net samples each user's slot-priority permutation from the
Plackett--Luce distribution
\begin{equation}
\label{eq:knet_priority_policy}
\begin{aligned}
p_{\theta}\!\left(
\boldsymbol{\pi}_{u}^{(\tau)}
\mid\mathsf{s}^{(\tau)}
\right)
=
\prod_{\ell=1}^{N}
\frac{
\exp\!\left(o_{u,\pi_{u,\ell}}^{(\tau)}/T_K\right)
}{
\sum_{j\in\mathcal J_{u,\ell}}
\exp\!\left(o_{u,j}^{(\tau)}/T_K\right)
},
\end{aligned}
\end{equation}
where $\mathcal J_{u,\ell}$ contains the unselected slots and $T_K$ is the
sampling temperature. At inference, the permutation is obtained by sorting
the logits in descending order.

Given $\boldsymbol{\pi}_{u}^{(\tau)}$, let $\mathcal R_{u,\ell}$ denote the
RBGs remaining before decision $\ell$. The sequential decoder assigns
\begin{equation}
\label{eq:knet_greedy_decoding}
\begin{aligned}
n_{u,\ell}^{\star}
&=
\arg\max_{n\in\mathcal R_{u,\ell}}
\xi_{u,\pi_{u,\ell},n}^{(\tau)},\\
k_{u,\pi_{u,\ell}}^{(\tau)}
&=
\begin{cases}
n_{u,\ell}^{\star},
&\xi_{u,\pi_{u,\ell},n_{u,\ell}^{\star}}^{(\tau)}>0,\\
\emptyset,
&\text{otherwise}.
\end{cases}
\end{aligned}
\end{equation}
Only an assigned RBG is removed from $\mathcal R_{u,\ell}$, ensuring a
feasible per-user one-to-one assignment. The network forward pass, sorting,
and sequential decoding require
$\mathcal O(UNK+UN\log N)$ operations, which reduces to
$\mathcal O(UK^2)$; this comparison concerns only the scheduling
step, whereas per-user Hungarian matching requires $\mathcal O(UK^3)$.

\subsubsection{Policy-Gradient Training}
K-Net maximizes the expected terminal WSR
$\mathcal J(\theta)
=\mathbb E_{\mathcal D,p_\theta}
[F(\mathbf V_{\theta}^{(T)},\mathbf k_{\theta}^{(T)})]$
using REINFORCE \cite{sutton1999policy}. For training sample $m$, the
FastFP--Hungarian reference uses the same channel, initialization, and horizon,
and defines the advantage
$A_m=
F(\mathbf V_{\theta,m}^{(T)},\mathbf k_{\theta,m}^{(T)})
-F(\mathbf V_{{\rm ref},m}^{(T)},\mathbf k_{{\rm ref},m}^{(T)})$.
Since the reference is independent of the sampled actions and $\theta$, the
policy-gradient estimator is
\begin{equation}
\label{eq:knet_reinforce_final}
\nabla_{\theta}\mathcal J(\theta)
=
\mathbb E_{\mathcal D,p_\theta}\!\left[
A_m\sum_{\tau=1}^{T}
\nabla_{\theta}
\log p_{\theta}\!\left(
a^{(\tau)}\mid\mathsf{s}^{(\tau)}
\right)
\right],
\end{equation}
The complete training procedure is summarized in
Algorithm~\ref{alg:bounded_training}.

\begin{algorithm}[t!]
\caption{Training of P-Net and K-Net}
\label{alg:bounded_training}
\begin{algorithmic}[1]
\STATE \textbf{input:} Training dataset $\mathcal D$ and unfolding horizon $T$.
\STATE Initialize P-Net parameters $\theta_{\rm P}$ and K-Net parameters $\theta_{\rm K}$.
\STATE \textbf{[P-Net Training with Fixed RBG Allocation]}
\REPEAT
    \STATE Sample a mini-batch $(\mathbf H, \mathbf V^{(0)}, \mathbf k^{(0)}) \sim\mathcal D$.
    \FOR{$\tau=1,\ldots,T$}
        \STATE Update $\mathbf\Gamma^{(\tau)}$ and $\mathbf Y^{(\tau)}$ via \eqref{eq:gamma_update} and \eqref{eq:y_update}.
        \STATE Extract $\mathbf s^{(\tau)}$, predict $\alpha^{(\tau)}\in(0,2)$, and update $\mathbf V^{(\tau)}$ via \eqref{eq:update_rule_pnet0}.
    \ENDFOR
    \STATE Update $\theta_{\rm P}$ by minimizing the negative terminal WSR in \eqref{eq:pnet_loss} using TBPTT.
\UNTIL{the P-Net training loss converges}

\STATE \textbf{[K-Net Training with FastFP]}
\REPEAT
    \STATE Sample a mini-batch $(\mathbf H, \mathbf V^{(0)}, \mathbf k^{(0)}) \sim\mathcal D$.
    \FOR{$\tau=1,\ldots,T$}
        \STATE Update the beamformers using FastFP and compute edge features $\mathbf E^{(\tau)}$.
        \STATE Sample slot priorities via \eqref{eq:knet_priority_policy} and update $\mathbf k^{(\tau)}$ using \eqref{eq:knet_greedy_decoding}.
    \ENDFOR
    \STATE Compute the terminal WSR reward and update $\theta_{\rm K}$ using \eqref{eq:knet_reinforce_final}.
\UNTIL{the K-Net training objective converges}
\STATE \textbf{output:} Trained parameters $\theta_{\rm P}^{*}$ and $\theta_{\rm K}^{*}$.
\end{algorithmic}
\end{algorithm}
\section{Numerical Results}
\label{sec:exp_results}
In this section, we first evaluate P-Net under fixed RBG assignments to isolate its beamforming acceleration capability. We then assess the complete P-Net and K-Net framework under joint RBG scheduling and beamforming optimization.
\subsection{Simulation Setup}
We consider a multi-cell downlink MU-MIMO network with a hexagonal cell layout, where each BS is located at the center of a cell, and users are uniformly distributed within the service area. All BSs share the same set of RBGs, and the distance between adjacent BSs is set to $D=0.8$ km. The number of data streams is identical for all users and satisfies $d_u = 2$. Each BS is equipped with $M_t = 32$ transmit antennas, and each user is equipped with $M_r = 4$ receive antennas. The maximum transmit power is $P_{\max}=40$ dBm. The noise power density is $N_0=-174$ dBm/Hz, and the bandwidth is $180$ kHz. The path loss is $PL(r) = 128.1 + 37.6 \log_{10}(r) + \xi$,
where $r$ is in kilometers and $\xi \sim \mathcal{N}(0,8^2)$ models log-normal shadowing. All user weights are $\rho_u=1$. Independent channel realizations are split into training and test sets. Within each comparison, all methods share the same RBG assignments and initial beamformers.

P-Net employs a single-layer GRU with 128 hidden units and is trained using Adam with a learning rate of $10^{-3}$, a batch size of 8, and TBPTT over $T=50$ iterations with truncation length $l=10$. K-Net consists of two-layer MLPs with 32 hidden units and is optimized using Adam with an initial learning rate of $10^{-4}$ and a gradually decaying learning-rate schedule for stable policy-gradient training.
All methods are implemented in Python 3.9.20 using PyTorch
2.5.1 with CUDA 11.8 and evaluated on a workstation equipped
with an NVIDIA GeForce RTX 4070 SUPER GPU and an AMD Ryzen 7 9700X CPU.

\subsection{Performance Evaluation under Fixed RBG Allocation}
To strictly isolate the contribution of the continuous beamforming optimizer, we first evaluate P-Net under fixed RBG assignments. The safety check is disabled to isolate the intrinsic acceleration effect of the learned relaxation factor.

\subsubsection{WSR Performance and Acceleration Ability}
We evaluate the beamforming performance of P-Net against representative iterative baselines. Specifically, we compare P-Net with FastFP, Nesterov-FastFP, and a recently proposed fixed-depth deep-unfolded FP method, DeepFP \cite{zhu2026deepfp}. Two variants of DeepFP are considered: DeepFP-Hybrid, which is first pretrained using teacher labels and then fine-tuned for weighted sum rate (WSR), and DeepFP-Unsupervised, which is trained directly to maximize the terminal WSR.
We first consider $(B,U,K)=(7,28,14)$ with 1000 channel realizations. After 50 iterations, P-Net averages $4096.58$ WSR, exceeding Nesterov-FastFP and FastFP by $13.55\%$ and $77.52\%$, respectively. The right-shifted CDF in Fig.~\ref{fig:perf0_cdf} shows that this gain persists across the sample-wise distribution.
\begin{figure}[!t]
    \centering
    \includegraphics[width=\columnwidth]{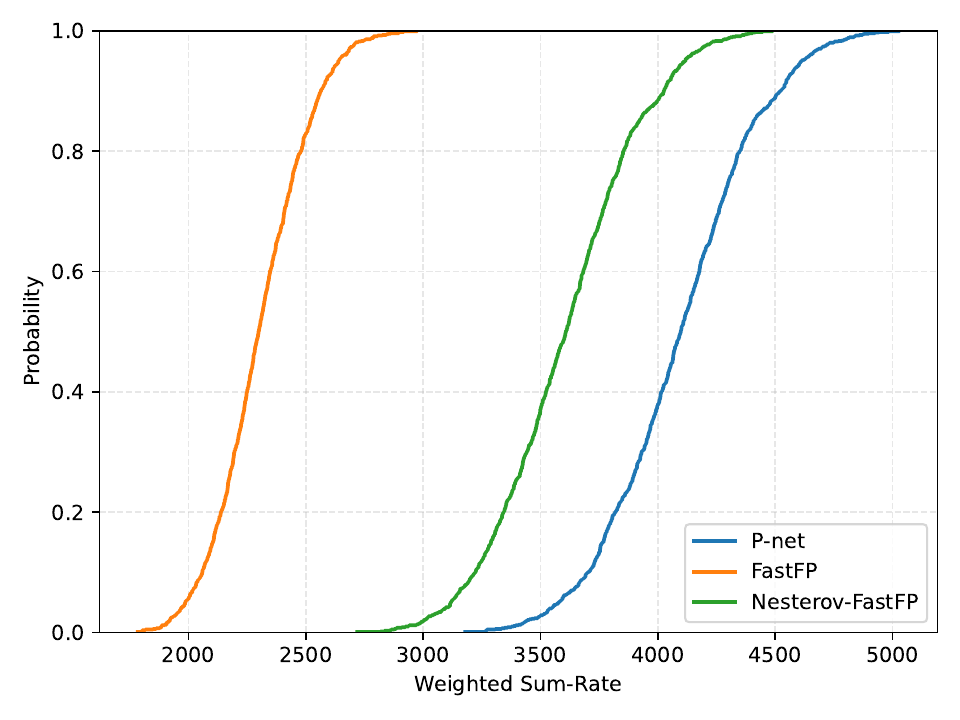}
    \caption{Final-WSR CDF after 50 iterations for $(B,U,K)=(7,28,14)$.}
    \label{fig:perf0_cdf}
\end{figure}

To jointly examine fixed-horizon performance and convergence speed, Table~\ref{tab:pnet_deepfp_runtime_merged} reports the average WSR, runtime, and time to target. For both DeepFP variants, we fix the unfolded depth to $T=8$ layers. This configuration is primarily chosen to maintain a comparable number of model parameters with P-Net, whereas P-Net is evaluated at both 8 and 50 iterations to characterize its recurrent parameter-sharing behavior. The reference target is the average WSR attained by Nesterov-FastFP after 50 iterations.

At $8$ iterations, P-Net attains the highest WSR ($2652.66$), followed by DeepFP-Unsupervised ($2380.75$) and DeepFP-Hybrid ($2068.75$). Their runtimes are comparable: $20.45$, $19.30$, and $20.10$ ms per CSI, respectively. With shared parameters, P-Net can continue beyond this fixed depth: at 50 iterations, it achieves $4096.58$, which is $13.55\%$ above Nesterov-FastFP, with only $0.56\%$ more full-horizon runtime. It reaches the Nesterov-FastFP target in 34 iterations and an estimated $79.32$ ms, reducing time to target by $31.62\%$; FastFP does not reach this target within 50 iterations. These results show that P-Net provides the strongest fixed-horizon WSR among the learned methods considered here and can continue improving beyond $8$ iterations while retaining model parameter sharing.

\begin{table*}[!t]
\caption{WSR and Runtime Comparison Under Fixed RBG Assignment for $(B,U,K)=(7,28,14)$}
\label{tab:pnet_deepfp_runtime_merged}
\centering
\small
\renewcommand{\arraystretch}{1.15}
\setlength{\tabcolsep}{4.5pt}
\begin{tabular}{lccccc}
\toprule
Method
& Inference Horizon
& Avg. WSR
& Runtime (ms/CSI)
& Iter. to Target
& Time to Target (ms) \\
\midrule
FastFP
& 50 iter.
& $2307.63$
& $103.59$
& $>50$
& -- \\

Nesterov-FastFP
& 50 iter.
& $3607.78$
& $116.01$
& 50
& $116.01$ (100\%) \\

DeepFP-Unsupervised
& 8 layers
& $2380.75$
& $19.30$
& --
& -- \\

DeepFP-Hybrid
& 8 layers
& $2068.75$
& $20.10$
& --
& -- \\

P-Net
& 8 iter.
& $\mathbf{2652.66}$
& $20.45$
& $>8$
& -- \\

P-Net
& 50 iter.
& $\mathbf{4096.58}$
& $116.65$
& $\mathbf{34}$
& $\mathbf{79.32}$ (68.38\%) \\
\bottomrule
\end{tabular}

\end{table*}

We next evaluate the denser $(B,U,K)=(7,35,20)$ setting on 1000 channel realizations. After 500 iterations, P-Net averages $9651.86$ WSR, $6.81\%$ above Nesterov-FastFP and $75.19\%$ above FastFP. The right-shifted histogram in Fig.~\ref{fig:perf1_dist} confirms the distribution-level gain.
\begin{figure}[!t]
    \centering
    \includegraphics[width=\columnwidth]{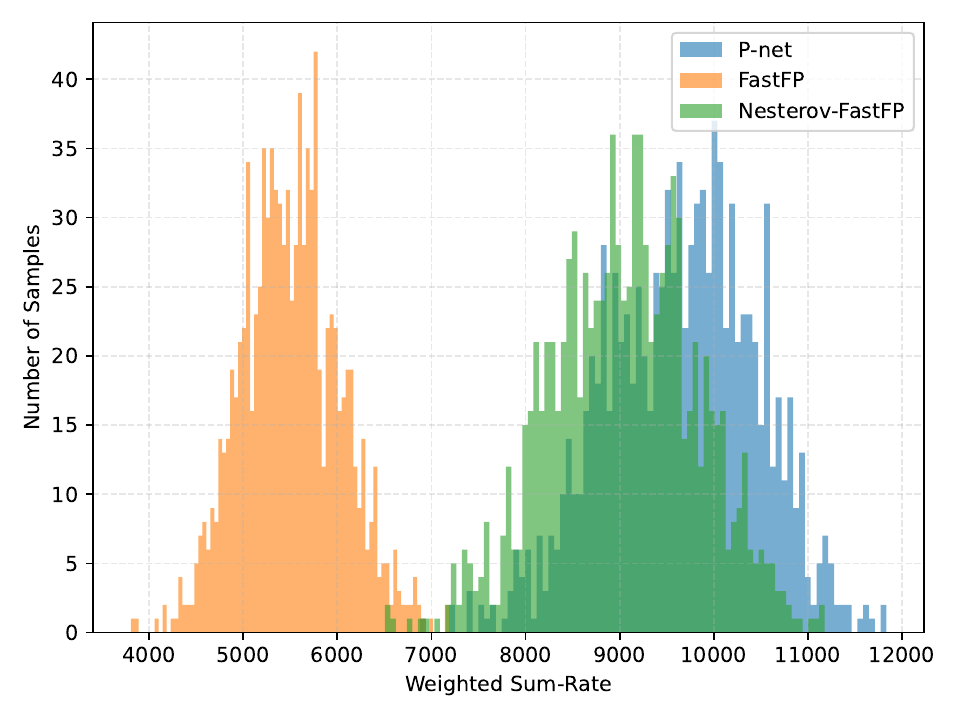}
    \caption{Final WSR distribution after 500 iterations for $(B,U,K)=(7,35,20)$.}
    \label{fig:perf1_dist}
\end{figure}

\begin{figure}[!t]
    \centering
    \includegraphics[width=\columnwidth]{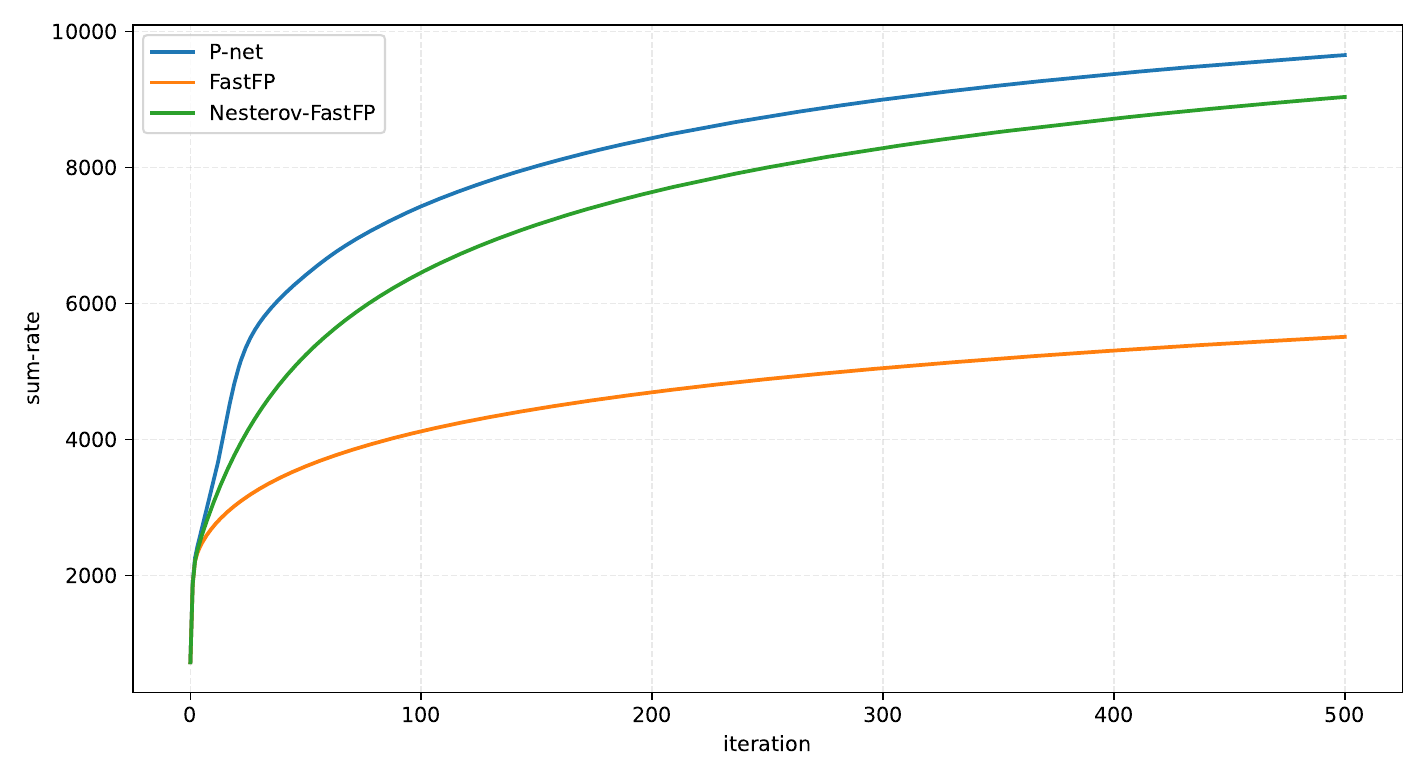}
\caption{Average WSR trajectories after 500 iterations for $(B,U,K)=(7,35,20)$.}
\label{fig:perf1_curve}
\end{figure}
To further evaluate the convergence behavior of the P-net, Fig.~\ref{fig:perf1_curve} shows that P-Net maintains a steeper WSR ascent than FastFP and Nesterov-FastFP over 500 iterations. As illustrated in Fig.~\ref{fig:alpha}, the learned relaxation factors remain within $(0,2)$ and are larger than the FastFP reference $\alpha=1$, with more aggressive values in early iterations and smaller values during later refinement. This behavior supports the interpretation of P-Net as an adaptive, bounded-step-size controller.

\begin{figure}[!t]
    \centering
    \includegraphics[width=\columnwidth]{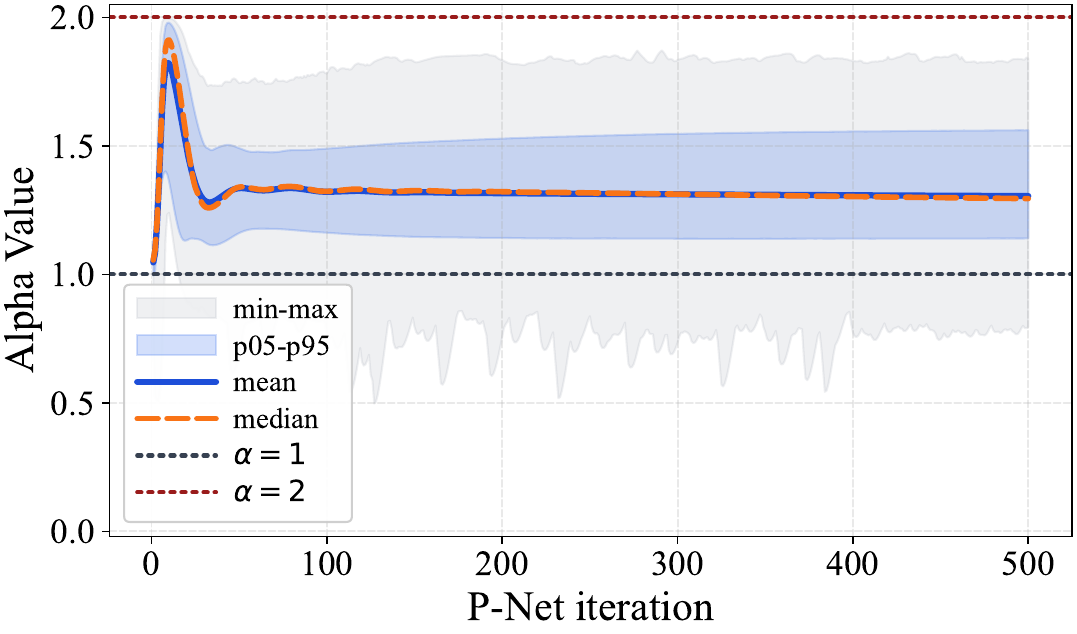}
    \caption{Statistical distribution (mean, median, and percentiles) of the learned relaxation factor $\alpha$ over 500 iterations for $(B,U,K)=(7,35,20)$.}
    \label{fig:alpha}
\end{figure}

\subsubsection{Generalization Under Network-Scale Variations}

Since P-Net operates on the physics features of the FastFP framework and shares parameters across user--RB links, it is expected to generalize more naturally across different values of $U$ and $K$ than a fully fixed-dimensional NN. We evaluate this property by training P-Net in the network with $B=7$ BSs, $U=28$ users, and $K=14$ RBGs, and then testing it under network settings with different numbers of users or RBGs. The first row of Fig.~\ref{fig:gen_rb1}--\ref{fig:gen_ue3} fixes $(B,U)=(7,28)$ and varies $K\in\{21,28,35\}$; the second fixes $(B,K)=(7,14)$ and varies $U\in\{35,42,49\}$. 

P-Net outperforms both baselines in all six cases. Relative to Nesterov-FastFP, its average-WSR gains are $36.16\%$--$36.45\%$ across RBG counts and $39.93\%$--$47.07\%$ across user counts, showing transfer across dimensions and a higher per-link power budget.

\begin{figure*}[!t]
    \centering
    \setlength{\tabcolsep}{2pt}
    \renewcommand{\arraystretch}{1.0}

    \begin{tabular}{ccc}
        \subfloat[$K=21$]{
            \includegraphics[width=0.31\textwidth]{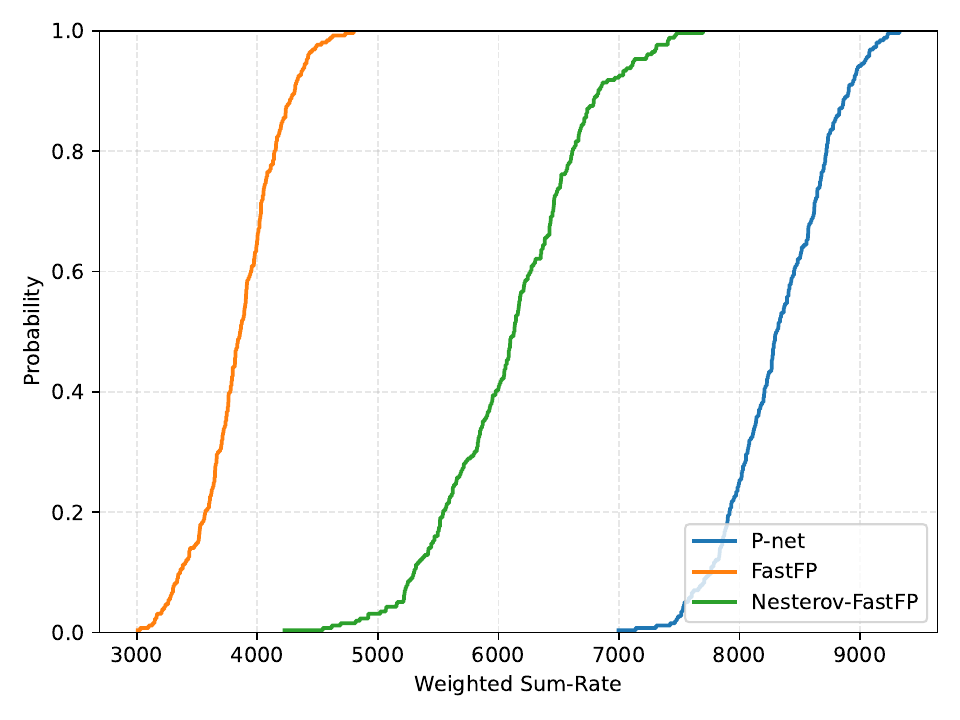}
            \label{fig:gen_rb1}
        }
        &
        \subfloat[$K=28$]{
            \includegraphics[width=0.31\textwidth]{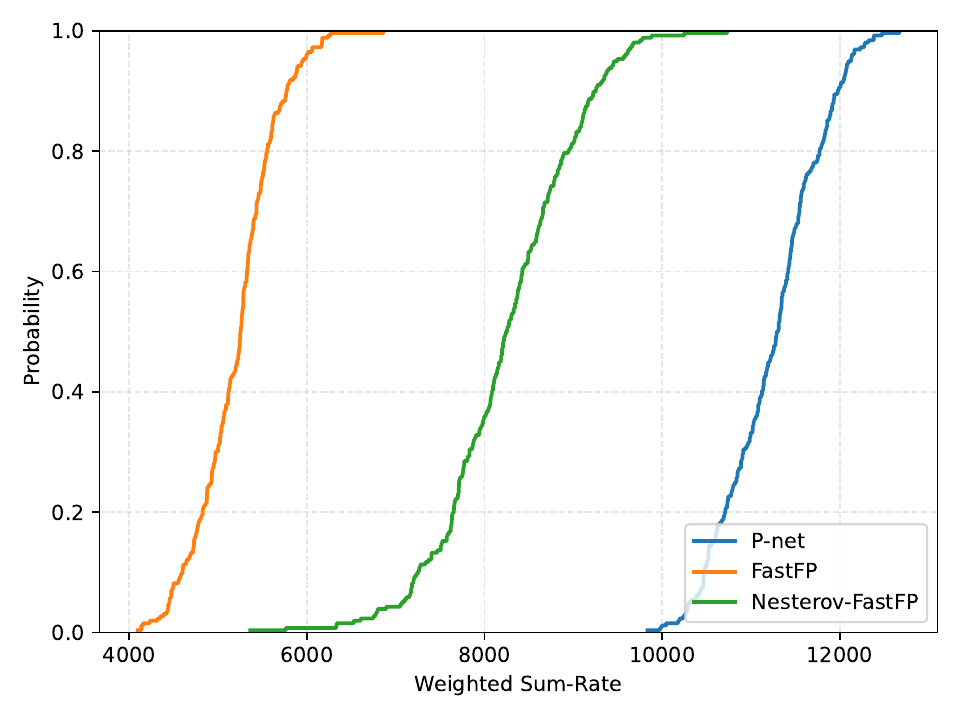}
            \label{fig:gen_rb2}
        }
        &
        \subfloat[$K=35$]{
            \includegraphics[width=0.31\textwidth]{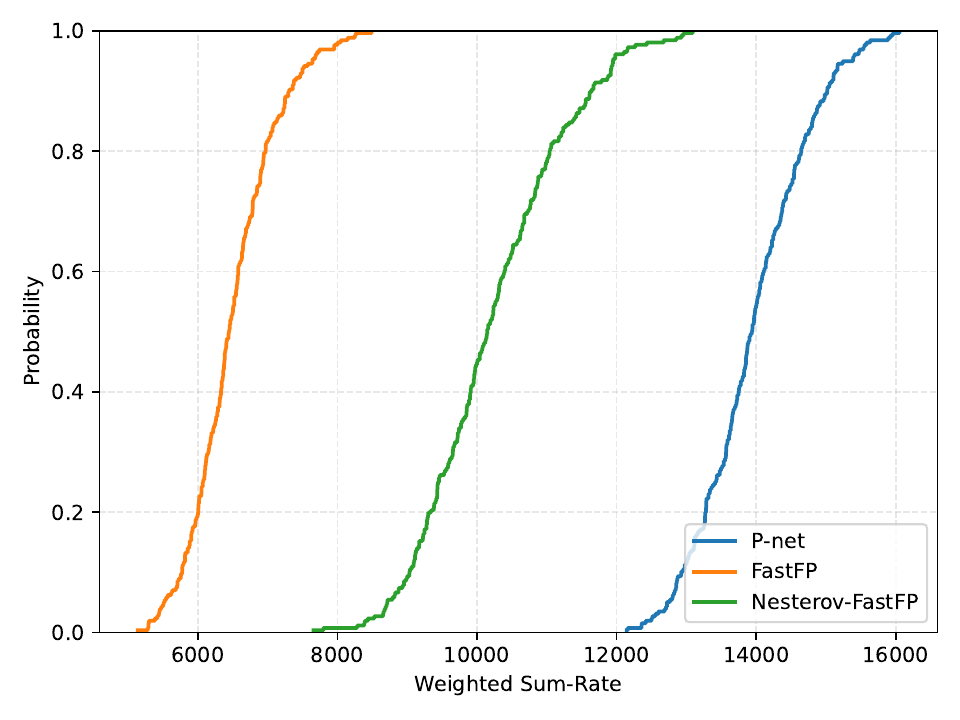}
            \label{fig:gen_rb3}
        }
        \\[+0.5em]
        \subfloat[$U=35$]{
            \includegraphics[width=0.31\textwidth]{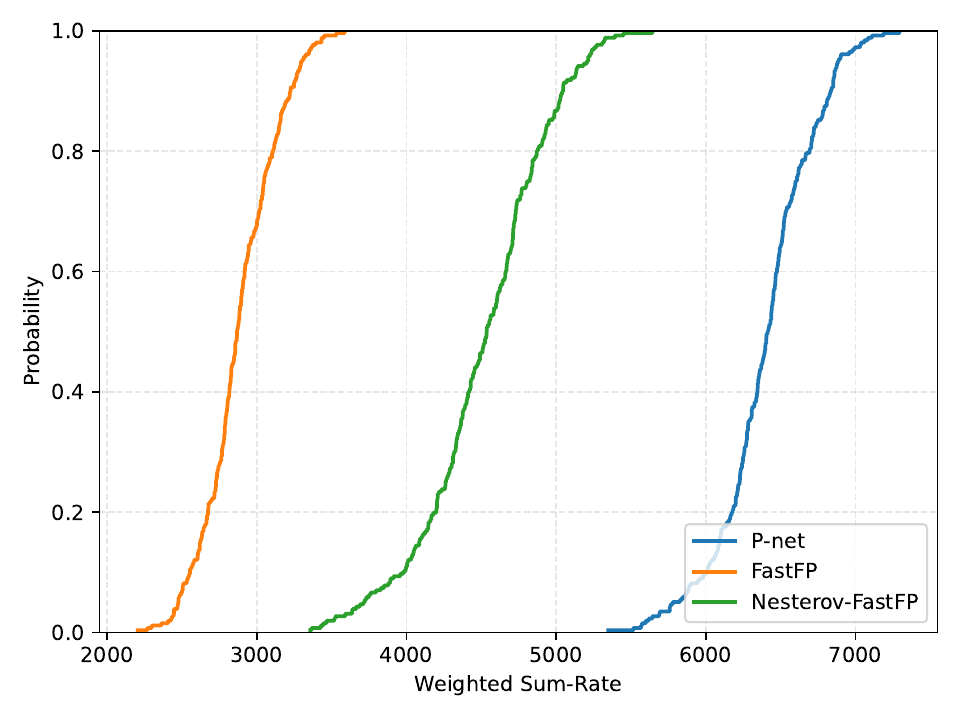}
            \label{fig:gen_ue1}
        }
        &
        \subfloat[$U=42$]{
            \includegraphics[width=0.31\textwidth]{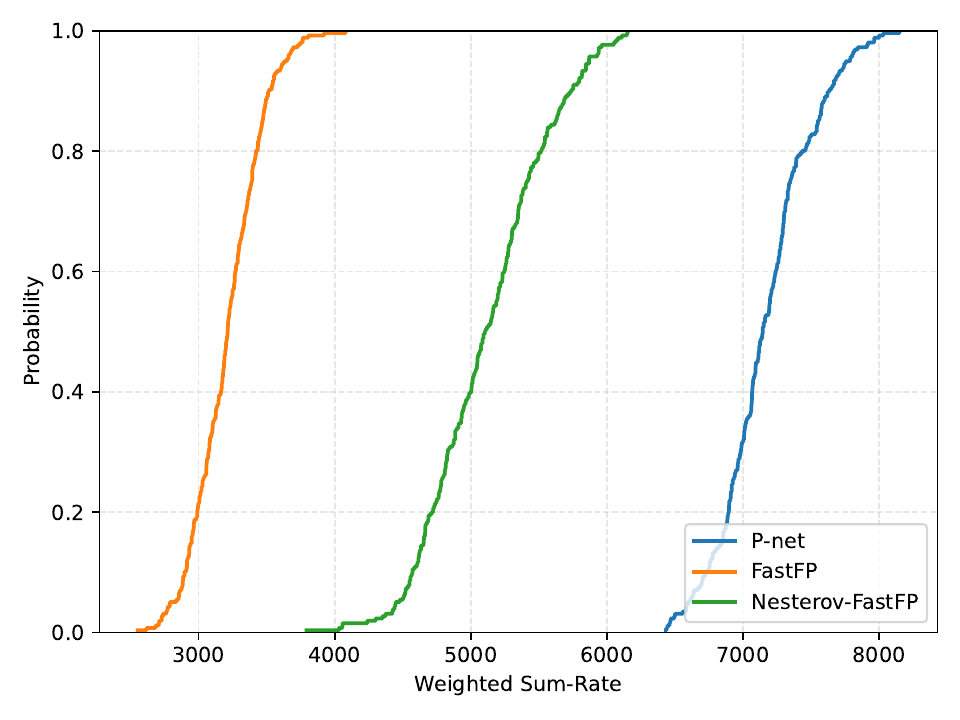}
            \label{fig:gen_ue2}
        }
        &
        \subfloat[$U=49$]{
            \includegraphics[width=0.31\textwidth]{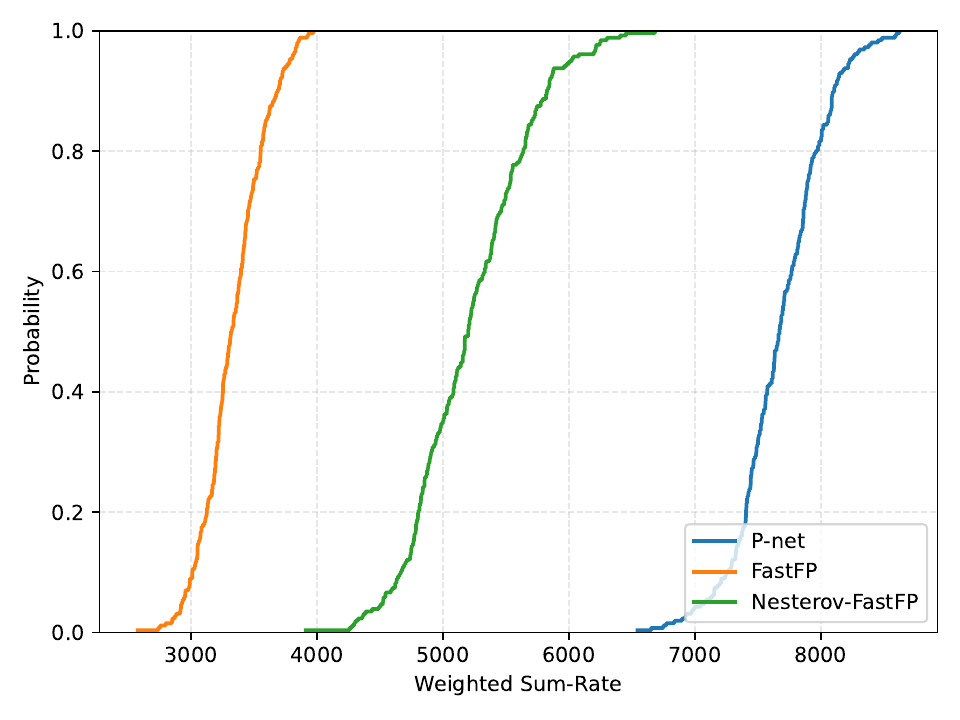}
            \label{fig:gen_ue3}
        }
    \end{tabular}

    \caption{WSR CDFs after 100 iterations under network-scale variations. The rows vary the numbers of RBGs and users, respectively.}
\end{figure*}

Table~\ref{tab:pnet_gen_shadow_power} evaluates the same P-Net model without retraining under power and shadowing shifts. P-Net outperforms Nesterov-FastFP in all nine cases, with gains increasing from $7.65\%$--$20.25\%$ at $20$ dBm to $26.84\%$--$47.55\%$ at $40$ dBm.

\begin{table}[!t]
\caption{Average WSR After 100 Iterations Under Transmit Power and Shadowing Shifts}
\label{tab:pnet_gen_shadow_power}
\centering
\footnotesize
\setlength{\tabcolsep}{3pt}
\renewcommand{\arraystretch}{1.08}
\begin{tabular}{@{}ccccc@{}}
\toprule
$P_{\mathrm{max}}$ (dBm) & $\sigma_{\mathrm{sh}}$ (dB) & P-Net & Nesterov-FastFP & Relative Gain \\
\midrule
\multirow{3}{*}{$20$}
& $4$  & $3899.58$ & $3622.47$ & $7.65\%$ \\
& $8$  & $3883.08$ & $3459.46$ & $12.25\%$ \\
& $12$ & $3848.02$ & $3200.07$ & $20.25\%$ \\
\midrule
\multirow{3}{*}{$30$}
& $4$  & $5239.81$ & $4421.75$ & $18.50\%$ \\
& $8$  & $5074.52$ & $4024.77$ & $26.08\%$ \\
& $12$ & $4828.85$ & $3531.88$ & $36.72\%$ \\
\midrule
\multirow{3}{*}{$40$}
& $4$  & $5883.89$ & $4638.82$ & $26.84\%$ \\
& $8$  & $5635.82$ & $4142.02$ & $36.06\%$ \\
& $12$ & $5261.12$ & $3565.62$ & $47.55\%$ \\
\bottomrule
\end{tabular}
\end{table}

\subsection{Performance Evaluation Under Joint Optimization of RBG Allocation and Beamforming}
\label{subsec:joint_vk_experiment}

The preceding experiments under fixed RBG assignments verify that P-Net can accelerate the continuous FastFP beamforming update. We next evaluate the complete alternating framework, where P-Net updates the beamformers and K-Net updates the RBG assignments, under $(B,U,K)=(7,14,28)$ and $(7,28,35)$. All methods use identical channel realizations and initializations. We compare FastFP + Hungarian, FastFP + Greedy, FastFP + K-Net, and P-Net + K-Net. The first three isolate the benefit of K-Net's learned decoding order over fixed-order greedy allocation, while the last comparison evaluates the additional acceleration provided by P-Net.

\begin{table*}[!t]
\caption{Performance and Runtime Summary for Joint RBG Allocation and Beamforming}
\label{tab:joint_vk_perf_two_scenarios}
\centering
\renewcommand{\arraystretch}{1.12}
\setlength{\tabcolsep}{3.2pt}
\resizebox{\textwidth}{!}{
\begin{tabular}{llcccccc}
\toprule
Scenario
& Method
& Final WSR
& WSR Gain
& CPU time (s)
& Time Ratio
& Speedup
& Target Iter./Time \\
\midrule
\multirow{4}{*}{$7$BS-$14$UE-$28$RB}
& FastFP + Hungarian
& $2730.70$
& Ref.
& $4.60$
& $100.00\%$
& $1.00\times$
& $100 / 4.60~(100.00\%)$ \\
& FastFP + Greedy
& $2603.77$
& $-4.65\%$
& $1.47$
& $31.86\%$
& $3.14\times$
& -- \\
& FastFP + K-Net
& $2874.88$
& $+5.28\%$
& $1.69$
& $36.74\%$
& $2.72\times$
& $67 / 1.13~(24.62\%)$ \\
& P-Net + K-Net
& $2929.13$
& $+7.27\%$
& $1.69$
& $36.72\%$
& $2.72\times$
& $53 / 0.89~(19.46\%)$ \\
\midrule
\multirow{4}{*}{$7$BS-$28$UE-$35$RB}
& FastFP + Hungarian
& $4670.34$
& Ref.
& $8.91$
& $100.00\%$
& $1.00\times$
& $100 / 8.91~(100.00\%)$ \\
& FastFP + Greedy
& $4516.82$
& $-3.29\%$
& $4.32$
& $48.46\%$
& $2.06\times$
& -- \\
& FastFP + K-Net
& $4722.97$
& $+1.13\%$
& $4.53$
& $50.83\%$
& $1.97\times$
& $93 / 4.21~(47.27\%)$ \\
& P-Net + K-Net
& $4936.90$
& $+5.71\%$
& $4.55$
& $51.02\%$
& $1.96\times$
& $67 / 3.05~(34.18\%)$ \\
\bottomrule
\end{tabular}}
\end{table*}

Table~\ref{tab:joint_vk_perf_two_scenarios} reports the final WSR and measured runtime per sample of the four joint optimization methods. In Scenario I, FastFP + Greedy reduces the runtime to $31.86\%$ of FastFP + Hungarian, but its final WSR decreases by $4.65\%$. This confirms that simply replacing Hungarian matching with a fixed-order greedy decoder provides acceleration at a clear performance cost. In contrast, FastFP + K-Net uses the same greedy decoder but replaces the default order with a learned priority order. It improves the final WSR by $5.28\%$ over FastFP + Hungarian while using only $36.74\%$ of the baseline runtime. Therefore, the gain of K-Net comes from learning a better decoding order rather than from changing the underlying FP-derived edge utility or invoking a more expensive assignment solver.

The same conclusion also holds in Scenario II. When the number of users increases from $14$ to $28$, FastFP + Greedy still remains below FastFP + Hungarian in final WSR. FastFP + K-Net, however, slightly exceeds the Hungarian baseline by $1.13\%$ while reducing the runtime to $50.83\%$. This result is important because the larger scenario increases the size and difficulty of the discrete scheduling problem, yet K-Net still preserves the WSR quality of Hungarian matching with much lower inference time.

The complete P-Net + K-Net method further strengthens this advantage in both scenarios. In Scenario I, it improves the WSR gain from $5.28\%$ to $7.27\%$ over FastFP + Hungarian and reaches the final Hungarian target using only $19.46\%$ of the baseline time. In Scenario II, the additional contribution of P-Net becomes more pronounced: P-Net + K-Net improves the final WSR by $5.71\%$ over FastFP + Hungarian and by $4.53\%$ over FastFP + K-Net. These results show that K-Net addresses the discrete-matching bottleneck, while P-Net further accelerates and improves the continuous beamforming trajectory after the allocation step has been changed.

\begin{figure}[!t]
    \centering
    \includegraphics[width=\columnwidth]{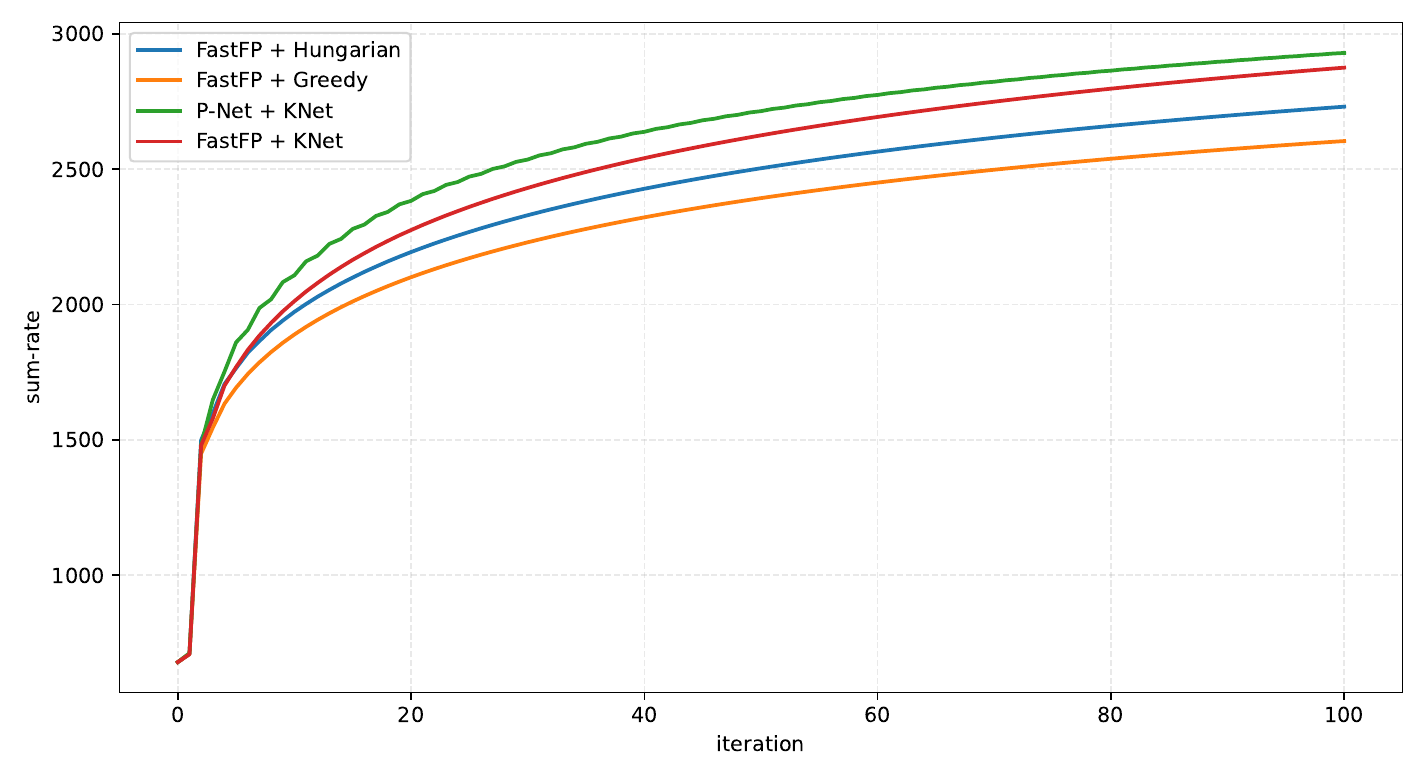}
    \caption{Average WSR trajectory of joint RBG allocation and beamforming optimization in Scenario I ($B=7$, $U=14$, and $K=28$).}
    \label{fig:joint_vk_trajectory_14ue28rb}
\end{figure}

\begin{figure}[!t]
    \centering
    \includegraphics[width=\columnwidth]{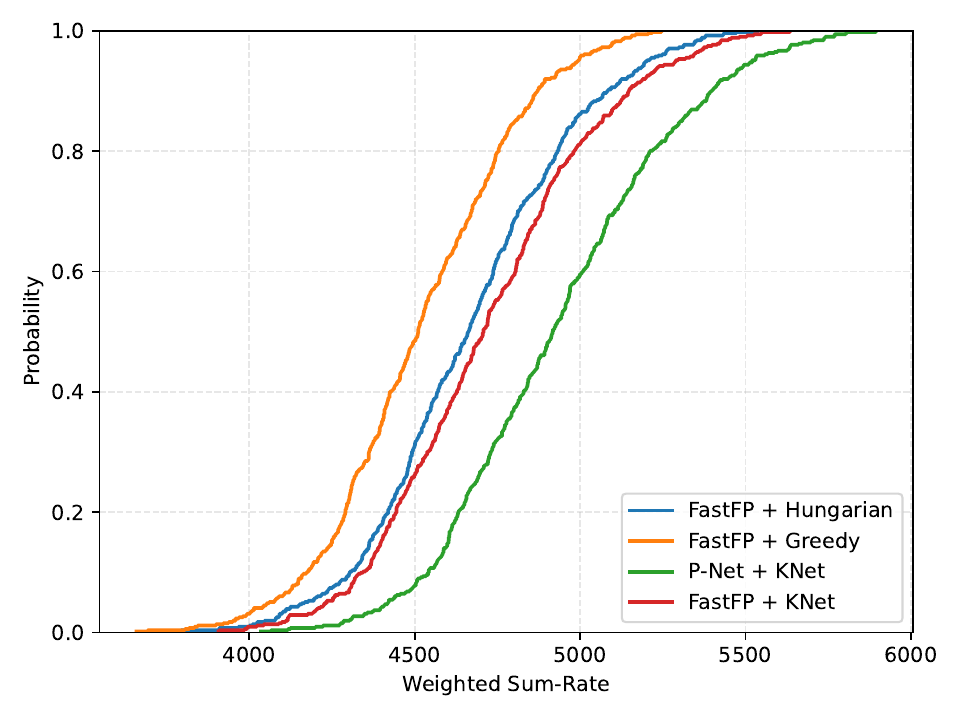}
    \caption{Empirical CDF of the final WSR for joint RBG allocation and beamforming optimization in Scenario II ($B=7$, $U=28$, and $K=35$).}
    \label{fig:joint_vk_cdf_28ue35rb}
\end{figure}

Fig.~\ref{fig:joint_vk_trajectory_14ue28rb} confirms that the K-Net-based methods reach the Hungarian reference substantially earlier in Scenario I, with P-Net + K-Net exhibiting the fastest ascent. The mild non-monotonicity after discrete reassignment is expected because K-Net optimizes the terminal WSR rather than enforcing improvement at every intermediate iteration. Fig.~\ref{fig:joint_vk_cdf_28ue35rb} provides the complementary distributional view for Scenario II: P-Net + K-Net yields the rightmost final-WSR distribution, while FastFP + K-Net remains close to the Hungarian reference, consistent with the average values in Table~\ref{tab:joint_vk_perf_two_scenarios}.

\section{Conclusion}
\label{sec:conclusion}
In this paper, we proposed a joint deep unfolding framework for mixed-integer RBG scheduling and beamforming in multi-cell MU-MIMO networks. P-Net learns a bounded relaxation factor to accelerate FastFP while preserving monotonic ascent and stationary-point convergence, whereas K-Net learns scheduling priorities for low-complexity greedy allocation, avoiding repeated Hungarian matching. Through recurrent parameter sharing, both modules support lightweight inference beyond the training horizon. Simulations demonstrate a superior WSR--runtime trade-off and robust generalization across unseen network scales and channel conditions.

\appendices

\section{Proof of Lemma~\ref{lemma:monotonicity}}
\label{app:proof_lemma1}

\begin{proof}
For fixed $\mathbf k$, let $Q_\tau(\mathbf V)$ denote the
FastFP minorant constructed at $\mathbf V^{(\tau-1)}$. It
satisfies
$Q_\tau(\mathbf V)\leq f(\mathbf V;\mathbf k)$ for all
$\mathbf V$, with equality at $\mathbf V^{(\tau-1)}$.
Its beamformer-dependent term for slot $(u,i)$ is
$q_{u,i}(\mathbf V_{u,i})
=2\Re\{\operatorname{Tr}(\mathbf V_{u,i}^{H}
\mathbf\Xi_{u,i}^{(\tau)})\}
-\lambda_{u,i}^{(\tau)}\|\mathbf V_{u,i}\|_F^2$.

Define
$\mathbf D_{u,i}^{(\tau)}
=\mathbf V_{u,i}^{(\tau)}
-\mathbf V_{u,i}^{(\tau-1)}$ and
$\mathbf G_{u,i}^{(\tau)}
=\mathbf\Xi_{u,i}^{(\tau)}
-\lambda_{u,i}^{(\tau)}
\mathbf V_{u,i}^{(\tau-1)}$.
The projected update is
$\mathbf V_{u,i}^{(\tau)}
=\mathcal P_{\mathcal W}
(\mathbf V_{u,i}^{(\tau-1)}
+\alpha_{u,i}^{(\tau)}
\mathbf G_{u,i}^{(\tau)}/\lambda_{u,i}^{(\tau)})$.
Since $\mathcal W$ is closed and convex, the first-order
optimality condition of the Euclidean projection yields:
\[
2\Re\!\left\{
\operatorname{Tr}\!\left(
\mathbf G_{u,i}^{(\tau)H}
\mathbf D_{u,i}^{(\tau)}
\right)\right\}
\geq
\frac{2\lambda_{u,i}^{(\tau)}}
{\alpha_{u,i}^{(\tau)}}
\|\mathbf D_{u,i}^{(\tau)}\|_F^2 .
\]
Using the quadratic form of $q_{u,i}$ therefore gives
\[
q_{u,i}(\mathbf V_{u,i}^{(\tau)})
-q_{u,i}(\mathbf V_{u,i}^{(\tau-1)})
\geq
\frac{\lambda_{u,i}^{(\tau)}
(2-\alpha_{u,i}^{(\tau)})}
{\alpha_{u,i}^{(\tau)}}
\|\mathbf D_{u,i}^{(\tau)}\|_F^2
\geq 0,
\]
where the last inequality follows from
$\alpha_{u,i}^{(\tau)}\in(0,2)$.
Summing over all active slots yields
$Q_\tau(\mathbf V^{(\tau)})
\geq Q_\tau(\mathbf V^{(\tau-1)})$.
The minorization and tightness properties then imply
$f(\mathbf V^{(\tau)};\mathbf k)
\geq f(\mathbf V^{(\tau-1)};\mathbf k)$, which proves the
claim.
\end{proof}

\section{Proof of Theorem~\ref{theorem:stationary}}
\label{app:proof_theorem1}
\begin{proof}
The proof proceeds by first establishing a uniform sufficient-ascent bound, which ensures objective convergence and vanishing successive differences. We then pass to the limit along an arbitrary convergent subsequence to obtain a projected fixed-point relation, whose first-order consistency with the original WSR objective establishes stationarity.

Let $\mathbf D_{u,i}^{(\tau)}
\triangleq
\mathbf V_{u,i}^{(\tau)}
-\mathbf V_{u,i}^{(\tau-1)}$.
The sufficient-ascent bound established in the proof of Lemma~\ref{lemma:monotonicity} gives
\begin{equation}
\begin{aligned}
&f(\mathbf V^{(\tau)};\mathbf k)
-f(\mathbf V^{(\tau-1)};\mathbf k)\\
&\quad\geq
\sum_{u,i}
\frac{\lambda_{u,i}^{(\tau)}
\left(2-\alpha_{u,i}^{(\tau)}\right)}
{\alpha_{u,i}^{(\tau)}}
\left\|\mathbf D_{u,i}^{(\tau)}\right\|_F^2
\geq
c\left\|\mathbf V^{(\tau)}
-\mathbf V^{(\tau-1)}\right\|_F^2
\end{aligned}
\label{eq:sufficient_ascent}
\end{equation}
for some constant $c>0$, owing to the uniform parameter
bounds.

The feasible set is compact, and the WSR is continuous and
upper-bounded. Hence, Lemma~\ref{lemma:monotonicity} implies
that the objective sequence converges. It then follows from
\eqref{eq:sufficient_ascent} that
$\|\mathbf V^{(\tau)}-\mathbf V^{(\tau-1)}\|_F\to0$.

Let $\mathbf V^\star$ be any accumulation point and choose a
subsequence such that
$\mathbf V^{(\tau_j)}\to\mathbf V^\star$.
The vanishing successive difference further implies
$\mathbf V^{(\tau_j+1)}\to\mathbf V^\star$.
After taking a further subsequence if necessary, the bounded
stepsizes
$\eta_{u,i}^{(\tau)}
\triangleq
\alpha_{u,i}^{(\tau)}/(2\lambda_{u,i}^{(\tau)})$
converge to some $\eta_{u,i}^\star>0$.

By the continuity of the auxiliary updates and Euclidean
projection, taking the limit in
\eqref{eq:update_rule_pnet0} yields
\begin{equation}
\mathbf V_{u,i}^\star
=
\mathcal P_{\mathcal W_{u,i}}
\left(
\mathbf V_{u,i}^\star
+
\eta_{u,i}^\star
\nabla_{\mathbf V_{u,i}}
Q(\mathbf V^\star\mid\mathbf V^\star)
\right).
\label{eq:limit_fixed_point}
\end{equation}
The exact auxiliary updates ensure the first-order consistency
$\nabla_{\mathbf V}Q(\mathbf V^\star\mid\mathbf V^\star)
=
\nabla_{\mathbf V}f(\mathbf V^\star;\mathbf k)$.
Therefore, the projection optimality condition associated with
\eqref{eq:limit_fixed_point} gives
\begin{equation}
\left\langle
\nabla_{\mathbf V}f(\mathbf V^\star;\mathbf k),
\mathbf V-\mathbf V^\star
\right\rangle
\leq 0,
\qquad
\forall\,\mathbf V\in\mathcal W,
\end{equation}
which is the first-order stationary condition for the
fixed-$\mathbf k$ beamforming subproblem. Since
$\mathbf V^\star$ was arbitrary, every accumulation point is
stationary.
\end{proof}

\bibliographystyle{IEEEtran}
\bibliography{unfolding_ref}

@techreport{itur_m2160_2023,
  title = {Framework and overall objectives of the future development of IMT for 2030 and beyond},
  author = {{ITU Radiocommunication Sector}},
  institution = {International Telecommunication Union},
  type = {Recommendation},
  number = {ITU-R M.2160-0},
  year = {2023},
  month = {November},
  url = {https://itu.int}
}

@techreport{ericsson_mobility_report_2024,
  title        = {Ericsson Mobility Report June 2024},
  institution  = {Ericsson},
  author      = {{Ericsson}},
  year         = {2024},
  month        = jun,
  url          = {https://www.ericsson.com/en/reports-and-papers/mobility-report}
}

@article{tataria2021sixg,
  author  = {Harsh Tataria and Mansoor Shafi and Andreas F. Molisch and Mischa Dohler and Henrik Sj{\"o}land and Fredrik Tufvesson},
  title   = {{6G} Wireless Systems: Vision, Requirements, Challenges, Insights, and Opportunities},
  journal = {Proceedings of the IEEE},
  year    = {2021},
  volume  = {109},
  number  = {7},
  pages   = {1166--1199},
  doi     = {10.1109/JPROC.2021.3061701}
}

@article{letaief2019roadmap,
  title={The roadmap to 6G: AI empowered wireless networks},
  author={Letaief, Khaled B and Chen, Wei and Shi, Yuanming and Zhang, Jun and Zhang, Ying-Jun Angela},
  journal={IEEE communications magazine},
  volume={57},
  number={8},
  pages={84--90},
  year={2019},
  publisher={IEEE}
}

@article{gesbert2010multicell,
  author  = {David Gesbert and Stephen Hanly and Howard Huang and Shlomo Shamai and Osvaldo Simeone and Wei Yu},
  title   = {Multi-Cell {MIMO} Cooperative Networks: A New Look at Interference},
  journal = {IEEE Journal on Selected Areas in Communications},
  year    = {2010},
  volume  = {28},
  number  = {9},
  pages   = {1380--1408},
  doi     = {10.1109/JSAC.2010.101202}
}

@article{bjornson2013optimal,
  author  = {Emil Bj{\"o}rnson and Eduard Jorswieck},
  title   = {Optimal Resource Allocation in Coordinated Multi-Cell Systems},
  journal = {Foundations and Trends in Communications and Information Theory},
  year    = {2013},
  volume  = {9},
  number  = {2--3},
  pages   = {113--381},
  doi     = {10.1561/0100000069}
}

@inproceedings{yu2011multicell,
  author    = {Wei Yu and Taesoo Kwon and Changyong Shin},
  title     = {Multicell Coordination via Joint Scheduling, Beamforming, and Power Spectrum Adaptation},
  booktitle = {Proceedings IEEE INFOCOM},
  year      = {2011},
  pages     = {2570--2578},
  doi       = {10.1109/INFCOM.2011.5935083}
}

@article{hassan2014downlink,
  author  = {Naveed Ul Hassan and Mohamad Assaad},
  title   = {Downlink Beamforming and Resource Allocation in Multicell {MISO-OFDMA} Systems},
  journal = {Transactions on Emerging Telecommunications Technologies},
  year    = {2014},
  volume  = {25},
  number  = {2},
  pages   = {173--182},
  doi     = {10.1002/ett.2554}
}

@inproceedings{khanafer2012mimo,
  author    = {Ali Khanafer and Teng Joon Lim and Roya Doostnejad and Taiwen Tang},
  title     = {{MIMO-OFDMA} Rate Allocation and Beamformer Design Using a Multi-Access Channel Framework},
  booktitle = {IEEE International Conference on Communications},
  year      = {2012},
  pages     = {2553--2558},
  doi       = {10.1109/ICC.2012.6364086}
}

@article{shen2018fp2,
  author  = {Kaiming Shen and Wei Yu},
  title   = {Fractional Programming for Communication Systems---Part {II}: Uplink Scheduling via Matching},
  journal = {IEEE Transactions on Signal Processing},
  year    = {2018},
  volume  = {66},
  number  = {10},
  pages   = {2631--2644},
  url     = {https://arxiv.org/abs/1802.10197}
}

@article{christensen2008wmmse,
  title={Weighted sum-rate maximization using weighted MMSE for MIMO-BC beamforming design},
  author={Christensen, S{\o}ren Skovgaard and Agarwal, Rajiv and De Carvalho, Elisabeth and Cioffi, John M},
  journal={IEEE Transactions on Wireless Communications},
  volume={7},
  number={12},
  pages={4792--4799},
  year={2008},
  publisher={IEEE}
}

@article{shi2011wmmse,
  author  = {Qingjiang Shi and Meisam Razaviyayn and Zhi-Quan Luo and Chen He},
  title   = {An Iteratively Weighted {MMSE} Approach to Distributed Sum-Utility Maximization for a {MIMO} Interfering Broadcast Channel},
  journal = {IEEE Transactions on Signal Processing},
  year    = {2011},
  volume  = {59},
  number  = {9},
  pages   = {4331--4340},
  doi     = {10.1109/TSP.2011.2147784}
}

@article{shen2018fp1,
  author  = {Kaiming Shen and Wei Yu},
  title   = {Fractional Programming for Communication Systems---Part {I}: Power Control and Beamforming},
  journal = {IEEE Transactions on Signal Processing},
  year    = {2018},
  volume  = {66},
  number  = {10},
  pages   = {2616--2630},
  doi     = {10.1109/TSP.2018.2812733}
}

@article{zhao2023rethinking,
  author  = {Xiaotong Zhao and Siyuan Lu and Qingjiang Shi and Zhi-Quan Luo},
  title   = {Rethinking {WMMSE}: Can Its Complexity Scale Linearly With the Number of {BS} Antennas?},
  journal = {IEEE Transactions on Signal Processing},
  year    = {2023},
  volume  = {71},
  pages   = {433--446},
  doi     = {10.1109/TSP.2023.3244104}
}

@article{pellaco2022matrixfree,
  title={Matrix-inverse-free deep unfolding of the weighted MMSE beamforming algorithm},
  author={Pellaco, Lissy and Bengtsson, Mats and Jald{\'e}n, Joakim},
  journal={IEEE Open Journal of the Communications Society},
  volume={3},
  pages={65--81},
  year={2021},
  publisher={IEEE}
}

@article{shen2024accelerating,
  author  = {Kaiming Shen and Ziping Zhao and Yannan Chen and Zepeng Zhang and Hei Victor Cheng},
  title   = {Accelerating Quadratic Transform and {WMMSE}},
  journal = {IEEE Journal on Selected Areas in Communications},
  year    = {2024},
  volume  = {42},
  number  = {11},
  pages   = {3110--3124},
  doi     = {10.1109/JSAC.2024.3431523}
}

@article{sun2017majorization,
  author  = {Ying Sun and Prabhu Babu and Daniel P. Palomar},
  title   = {Majorization-Minimization Algorithms in Signal Processing, Communications, and Machine Learning},
  journal = {IEEE Transactions on Signal Processing},
  year    = {2017},
  volume  = {65},
  number  = {3},
  pages   = {794--816},
  doi     = {10.1109/TSP.2016.2601299}
}

@article{sun2018learning,
  author  = {Haoran Sun and Xiangyi Chen and Qingjiang Shi and Mingyi Hong and Xiao Fu and Nicholas D. Sidiropoulos},
  title   = {Learning to Optimize: Training Deep Neural Networks for Wireless Resource Management},
  journal = {IEEE Transactions on Signal Processing},
  year    = {2018},
  volume  = {66},
  number  = {20},
  pages   = {5438--5453},
  doi     = {10.1109/TSP.2018.2866382}
}

@article{liang2019deep,
  title={Deep-learning-based wireless resource allocation with application to vehicular networks},
  author={Liang, Le and Ye, Hao and Yu, Guanding and Li, Geoffrey Ye},
  journal={Proceedings of the IEEE},
  volume={108},
  number={2},
  pages={341--356},
  year={2019},
  publisher={IEEE}
}

@article{monga2021algorithm,
  author  = {Vishal Monga and Yuelong Li and Yonina C. Eldar},
  title   = {Algorithm Unrolling: Interpretable, Efficient Deep Learning for Signal and Image Processing},
  journal = {IEEE Signal Processing Magazine},
  year    = {2021},
  volume  = {38},
  number  = {2},
  pages   = {18--44},
  doi     = {10.1109/MSP.2020.3016905}
}

@article{hu2021iaidnn,
  author  = {Qiyu Hu and Yunlong Cai and Qingjiang Shi and Kaidi Xu and Guanding Yu and Zhi Ding},
  title   = {Iterative Algorithm Induced Deep-Unfolding Neural Networks: Precoding Design for Multiuser {MIMO} Systems},
  journal = {IEEE Transactions on Wireless Communications},
  year    = {2021},
  volume  = {20},
  number  = {2},
  pages   = {1394--1410},
  doi     = {10.1109/TWC.2020.3033334}
}

@article{chowdhury2024uwmmse,
  author  = {Arindam Chowdhury and Gunjan Verma and Ananthram Swami and Santiago Segarra},
  title   = {Deep Graph Unfolding for Beamforming in {MU-MIMO} Interference Networks},
  journal = {IEEE Transactions on Wireless Communications},
  year    = {2024},
  volume  = {23},
  number  = {5},
  pages   = {4889--4903},
  doi     = {10.1109/TWC.2023.3323207}
}

@article{zhu2026deepfp,
  title={DeepFP: Deep-Unfolded Fractional Programming for MIMO Beamforming},
  author={Zhu, Jianhang and Chang, Tsung-Hui and Xiang, Liyao and Shen, Kaiming},
  journal={IEEE Transactions on Communications},
  year={2026},
  publisher={IEEE}
}

@inproceedings{liu2019alista,
  title={ALISTA: Analytic weights are as good as learned weights in LISTA},
  author={Liu, Jialin and Chen, Xiaohan and Wang, Zhangyang and Yin, Wotao},
  booktitle={International conference on learning representations},
  year={2019}
}

@article{bertocchi2020irestnet,
  author  = {Carla Bertocchi and Emilie Chouzenoux and Marie-Caroline Corbineau and Jean-Christophe Pesquet and Marco Prato},
  title   = {Deep Unfolding of a Proximal Interior Point Method for Image Restoration},
  journal = {Inverse Problems},
  year    = {2020},
  volume  = {36},
  number  = {3},
  pages   = {034005},
  doi     = {10.1088/1361-6420/ab460a}
}

@article{mukherjee2023guarantees,
  author  = {Subhadip Mukherjee and Andreas Hauptmann and Ozan {\"O}ktem and Marcelo Pereyra and Carola-Bibiane Sch{\"o}nlieb},
  title   = {Learned Reconstruction Methods With Convergence Guarantees: A Survey of Concepts and Applications},
  journal = {IEEE Signal Processing Magazine},
  year    = {2023},
  volume  = {40},
  number  = {1},
  pages   = {164--182},
  doi     = {10.1109/MSP.2022.3207451}
}

@article{huang2009downlink,
  title={Downlink Scheduling and Resource Allocation for OFDM Systems},
  author={Huang, Jianwei and Subramanian, Vijay G. and Agrawal, Rajeev and Berry, Randall A.},
  journal={IEEE Transactions on Wireless Communications},
  volume={8},
  number={1},
  pages={288--296},
  year={2009},
  doi={10.1109/T-WC.2009.070080}
}

@article{liu2014complexity,
  title={On the Complexity of Joint Subcarrier and Power Allocation for Multi-User OFDMA Systems},
  author={Liu, Ya-Feng and Dai, Yu-Hong},
  journal={IEEE Transactions on Signal Processing},
  volume={62},
  number={3},
  pages={583--596},
  year={2014},
  doi={10.1109/TSP.2013.2293123}
}

@article{song2008weighted,
  title={Weighted Max-Min Fair Beamforming, Power Control, and Scheduling for a MISO Downlink},
  author={Song, Bongyong and Lin, Yih-Hao and Cruz, Rene L.},
  journal={IEEE Transactions on Wireless Communications},
  volume={7},
  number={2},
  pages={464--469},
  year={2008},
  doi={10.1109/TWC.2008.060878}
}

@article{khan2020fp_hungarian,
  title={Optimizing Downlink Resource Allocation in Multiuser MIMO Networks via Fractional Programming and the Hungarian Algorithm},
  author={Khan, Ahmad Ali and Adve, Raviraj S. and Yu, Wei},
  journal={IEEE Transactions on Wireless Communications},
  volume={19},
  number={8},
  pages={5162--5175},
  year={2020},
  doi={10.1109/TWC.2020.2990126}
}

@inproceedings{lu2020learning,
  title={Learning-based massive beamforming},
  author={Lu, Siyuan and Zhao, Shengiie and Shi, Qingjiang},
  booktitle={GLOBECOM 2020-2020 IEEE Global Communications Conference},
  pages={1--6},
  year={2020},
  organization={IEEE}
}

@article{shen2020lorm,
  title={LORM: Learning to Optimize for Resource Management in Wireless Networks With Few Training Samples},
  author={Shen, Yifei and Shi, Yuanming and Zhang, Jun and Letaief, Khaled B.},
  journal={IEEE Transactions on Wireless Communications},
  volume={19},
  number={1},
  pages={665--679},
  year={2020},
  doi={10.1109/TWC.2019.2948132}
}

@article{alwarafy2021drlsurvey,
  title={Deep Reinforcement Learning for Radio Resource Allocation and Management in Next Generation Heterogeneous Wireless Networks: A Survey},
  author={Alwarafy, Abdulmalik and Abdallah, Mohamed and Ciftler, Bekir Sait and Al-Fuqaha, Ala and Hamdi, Mounir},
  journal={arXiv preprint arXiv:2106.00574},
  year={2021},
  url={https://arxiv.org/abs/2106.00574}
}

@article{cui2019spatial,
  title={Spatial Deep Learning for Wireless Scheduling},
  author={Cui, Wei and Shen, Kaiming and Yu, Wei},
  journal={IEEE Journal on Selected Areas in Communications},
  volume={37},
  number={6},
  pages={1248--1261},
  year={2019},
  doi={10.1109/JSAC.2019.2904322}
}

@article{shen2021gnn_rrm,
  title={Graph Neural Networks for Scalable Radio Resource Management: Architecture Design and Theoretical Analysis},
  author={Shen, Yifei and Shi, Yuanming and Zhang, Jun and Letaief, Khaled B.},
  journal={IEEE Journal on Selected Areas in Communications},
  volume={39},
  number={1},
  pages={101--115},
  year={2021},
  doi={10.1109/JSAC.2020.3036965}
}

@article{wang2022decentralized_gnn,
  title={Learning Decentralized Wireless Resource Allocations With Graph Neural Networks},
  author={Wang, Zhiyang and Eisen, Mark and Ribeiro, Alejandro},
  journal={IEEE Transactions on Signal Processing},
  volume={70},
  pages={1850--1863},
  year={2022},
  doi={10.1109/TSP.2022.3157316}
}

@article{sutton1999policy,
  title={Policy gradient methods for reinforcement learning with function approximation},
  author={Sutton, Richard S and McAllester, David and Singh, Satinder and Mansour, Yishay},
  journal={Advances in neural information processing systems},
  volume={12},
  year={1999}
}

\end{document}